\documentclass[12pt]{article}

\usepackage{wide}
\usepackage{amsmath,graphicx}
\usepackage[ruled]{algorithm}
\usepackage{algorithmicx}
\usepackage{algpseudocode}
\usepackage{amssymb}
\alglanguage{pseudocode}
\floatstyle{ruled}
\newfloat{algo}{htpb}{lop}
\floatname{algo}{Algorithm}

\algblockdefx{If}{EndIf}[1]{%
        {\bf If}\ #1\ {\bf then}}{{\bf End If}}
\algcblock[If]{If}{Else}{End If}
\algblockdefx{For}{EndFor}[3]{%
        {\bf for}\ #1$=$#2\ \textbf{to} #3\ {\bf do}}{{\bf done}}
\algblockdefx{DownFor}{EndDownFor}[3]{%
        {\bf for}\ #1$=$#2\ \textbf{down to} #3\ {\bf do}}{{\bf done}}

\algblockdefx{While}{EndWhile}[1]{%
        {\bf while}\ #1\ {\bf do}}{{\bf done}}
\newcommand{\trans}[1]{{\tt T}(#1)}

\newcommand{\w}[1]{\ensuremath{\textit{Write}(#1)}}
\newcommand{\pr}[1]{\ensuremath{G\textit{-Read}(#1)}}
\newcommand{\precA}{\ensuremath{\prec_{\mathcal{A}}}}
  \newcommand{\prioA}{\ensuremath{\triangleleft_{\mathcal{A}}}}
\newcommand{\GC}{\ensuremath{\mathbf{GC}}} 
\newcommand{\FA}{\ensuremath{A}}
\newcommand{\MG}{\ensuremath{M}}
\newcommand{\HG}{\ensuremath{H}}
\newcommand{\neigh}{\ensuremath{\Gamma}}
\newcommand{\DeltaG}{\ensuremath{\Delta}}
\newcommand{\DeltamGC}{\ensuremath{\mathbf{d}}}
\newcommand{\HGC}{\ensuremath{\mathfrak{H}}}

\newcommand{\suchthat}{\ : \ }
\usepackage{amsthm}
\usepackage{dsfont}
\newtheorem{remark}{Remark}
\newtheorem{lemma}{Lemma}
\newtheorem{corollary}{Corollary}
\newtheorem{theorem}{Theorem}
\newtheorem{definition}{Definition}

\newcommand{\step}{\mapsto}

\renewcommand{\to}{\longmapsto}

\title{Acyclic Strategy for Silent Self-Stabilization in Spanning
  Forests \footnote{This study has been partially supported by the \textsc{anr}
  projects \textsc{Descartes} (ANR-16-CE40-0023) and \textsc{Estate}
  (ANR-16-CE25-0009).}}

\author{\textbf{Karine Altisen, St\'ephane Devismes} \\
{\small\normalfont  Univ. Grenoble Alpes, CNRS, Grenoble
  INP\footnote{Institute of Engineering Univ. Grenoble Alpes},
  VERIMAG, 38000 Grenoble, 
  France } \\ 
{\small\normalfont  Firstname.Lastname@univ-grenoble-alpes.fr}
\\
\textbf{Ana\"is Durand} \\
{\small\normalfont IRISA, Université de Rennes, 35042 Rennes, France}\\
{\small\normalfont Anais.Durand@inria.fr}}

\date{}

\begin{document}

\maketitle

\begin{abstract}
In this paper, we formalize design patterns, commonly used in the
self-stabilizing area, to obtain general statements regarding both
correctness and time complexity guarantees.
Precisely, we study a general class of algorithms designed for
networks endowed with a sense of direction describing a spanning
forest ({\em e.g.}, a directed tree or a network where a directed
spanning tree is available) whose characterization is a simple ({\em
  i.e.}, quasi-syntactic) condition. We show that any algorithm of
this class is (1) silent and self-stabilizing under the distributed
unfair daemon, and (2) has a stabilization time which is polynomial in
moves and asymptotically optimal in rounds. To illustrate the
versatility of our method, we review several existing works where our
results apply.

\bigskip

\paragraph{Keywords:} Self-stabilization, silence, tree networks,
bottom-up actions, and top-down actions.
\end{abstract}

\section{Introduction}

\textit{Self-stabilization}~\cite{D74j} is a versatile technique to
withstand \textit{any finite number} of transient faults in a
distributed system: regardless of the \emph{arbitrary} initial
configuration of the system (and therefore also after the occurrence
of transient faults), a self-stabilizing (distributed) algorithm is
able to recover in finite time a so-called {\em legitimate}
configuration from which its behavior conforms to its specification.

After the seminal work of Dijkstra, many self-stabilizing algorithms
have been proposed to solve various tasks such as spanning tree
constructions~\cite{BPRT10}, token circulations~\cite{HuangC93}, clock
synchronization~\cite{CouvreurFG92}, propagation of information with
feedbacks~\cite{BuiDPV99}.  Those works consider a large taxonomy of
topologies: rings~\cite{MasuzawaK05,BlinT18}, (directed)
trees~\cite{BuiDPV99,chaudhuri-labelling,kdomset-turau}, planar
graphs~\cite{LinC10,GhoshK93}, arbitrary connected
graphs~\cite{DattaGPV01,ACDDP14}, {\em etc.}
Among those topologies, the class of directed (in-) trees ({\em i.e.},
trees where one process is distinguished as the root and edges are
oriented toward the root) is of particular interest.
Indeed, such topologies often appear, at an intermediate level, in
self-stabilizing composite algorithms.  {\em Composition} is a popular
way to design self-stabilizing algorithms~\cite{T01} since it allows
to simplify both the design and proofs.  Numerous self-stabilizing
algorithms, {\em e.g.},~\cite{AGH90,BPRT10,DDHLR16j}, are actually
made as a composition of a spanning directed treelike ({\em e.g.} tree
or forest) construction and some other algorithms specifically
designed for directed tree/forest topologies.
Notice that, even though not mandatory, most of these constructions
achieve an additional property called {\em silence}~\cite{DolevGS96}:
a silent self-stabilizing algorithm converges within finite time to a
configuration from which the values of the communication registers
used by the algorithm remain fixed.
Silence is a desirable property. Indeed, as noted in~\cite{DolevGS96},
the silent property usually implies more simplicity in the algorithm
design, and so allows to write simpler proofs; moreover, a silent
algorithm may utilize less communication operations and communication
bandwidth.

In this paper, we consider the locally shared memory model with
composite atomicity introduced by Dijkstra~\cite{D74j}, which is the
most commonly used model in self-stabilization.  In this model,
executions proceed in (atomic) steps and the asynchrony of the system
is captured by the notion of {\em daemon}. The weakest ({\em i.e.},
the most general) daemon is the {\em distributed unfair
  daemon}. Hence, solutions stabilizing under such an assumption are
highly desirable, because they work under any other daemon assumption.

The daemon assumption and time complexity are closely related. The
{\em stabilization time}, {\em i.e.}, the maximum time to reach a
legitimate configuration starting from an arbitrary one, is the main
time complexity measure to compare self-stabilizing algorithms.  It is
usually evaluated in terms of rounds, which capture the execution time
according to the speed of the slowest process. But, another crucial
issue is the number of local state updates, called {\em
  moves}. Indeed, the stabilization time in moves captures the amount
of computations an algorithm needs to recover a correct behavior.
Now, this latter complexity can be bounded only if the algorithm works
under an unfair daemon.  Actually, if an algorithm requires a stronger
daemon to stabilize, {\em e.g.}, a {\em weakly fair} daemon, then it
is possible to construct executions whose convergence is arbitrary
long in terms of (atomic) steps, meaning that, in such executions,
there are processes whose moves do not make the system progress in the
convergence. In other words, these latter processes waste computation
power and so energy.  Such a situation should be therefore prevented,
making the unfair daemon more desirable than the weakly fair one.

There are many self-stabilizing algorithms proven under the
distributed unfair daemon, {\em
  e.g.},~\cite{ACDDP14,CDDLR15,DLP11,DLV11,GHIJ14}.  However, analyses
of the stabilization time in moves is rather unusual and this may be
an important issue.
Indeed, recently, several self-stabilizing algorithms which work under
a distributed unfair daemon have been shown to have an exponential
stabilization time in moves in the worst case, {\em e.g.}, the silent
leader election algorithms from~\cite{DLP11,DLV11} (as shown
in~\cite{ACDDP14}), the Breadth-First Search (BFS) algorithm of Huang
and Chen~\cite{HC92} (as shown in~\cite{DJ16}), or the silent
self-stabilizing algorithm for the shortest-path spanning tree
of~\cite{GHIJ14} (as shown in~\cite{GHIJ16}).

\paragraph{Contribution.}

In this paper, we formalize design patterns, commonly used in the
self-stabilizing area, to obtain general statements regarding both
correctness and time complexity guarantees. Precisely, we study a general class
of algorithms designed for networks endowed with a sense of direction
describing a spanning forest ({\em e.g.}, a directed tree, or a
network where a directed spanning tree is available) whose
characterization is a simple ({\em i.e.}, quasi-syntactic)
condition. We show that any algorithm of this class is (1) silent and
self-stabilizing under the distributed unfair daemon, and (2) has a
stabilization time which is polynomial in moves and asymptotically
optimal in rounds.

Our condition, referred to as {\em acyclic strategy}, is based on the
notions of {\em top-down} and {\em bottom-up} actions.  Until now,
these types of actions was used rather informally in the context of
self-stabilizing algorithms dedicated to directed trees. Our first
goal has been to formally define these two paradigms. We have then
compiled this formalization together with a notion of acyclic
causality between actions and a last criteria called {\em
  correct-alone} ({\em n.b.}, only this latter criteria is not
syntactic) to obtain the notion of {\em acyclic strategy}. We show
that any algorithm that follows an acyclic strategy reaches a terminal
configuration in a polynomial number of moves, assuming a distributed
unfair daemon. Hence, if its terminal configurations conform to the
specification, then the algorithm is both silent and
self-stabilizing. Unfortunately, we show that our condition is not
sufficient to guarantee a stabilization time that is asymptotically
optimal in rounds, {\em i.e.}, $O(\HG)$ rounds where $\HG$ is the
height of the spanning forest.  However, we propose to enforce our
condition with an extra property, called {\em local mutual
  exclusivity}, which is sufficient to obtain the asymptotic optimal
bound in rounds. Finally, we propose a generic method to add this
latter property to any algorithm that follows an acyclic strategy but
is not locally mutually exclusive, allowing then to obtain a
complexity in $O(\HG)$ rounds. Our method has no overhead in terms of
moves. Finally, to illustrate the versatility of our method, we review
several existing works where our results apply.

\paragraph{Related Work.}

General schemes and efficiency are usually understood as orthogonal
issues.  For example, general
schemes have been proposed~\cite{KP93j,boldi_universal_2002} to
transform almost any algorithm (specifically, those algorithms that
can be self-stabilized) for arbitrary connected and identified
networks into their corresponding stabilizing version.  Such universal
transformers are, by essence, inefficient both in terms of space and
time complexities: their purpose is only to demonstrate the
feasibility of the transformation.  In~\cite{KP93j}, authors consider
asynchronous message-passing systems, while the synchronous locally
shared memory model is assumed in~\cite{boldi_universal_2002}.

However, few works, like~\cite{DT01j,DDT06j,BlinFP14}, target both
general self-stabilizing algorithm patterns and efficiency in rounds.

In \cite{DT01j,DDT06j}, authors propose a method to design silent
self-stabilizing algorithms for a class of fix-point problems (namely
fix-point problems which can be expressed using $r$-operators).  Their
solution works in non-bidirectional networks using bounded memory per
process. In~\cite{DT01j}, they consider the locally shared memory
model with composite atomicity assuming a distributed unfair daemon,
while in~\cite{DDT06j}, they bring their approach to asynchronous
message-passing systems. In both papers, they establish a
stabilization time in $O(D)$ rounds, where $D$ is the network
diameter, that holds for the synchronous case only, moreover move
complexity is not considered.

The remainder of the related work only concerns the locally shared
memory model with composite atomicity assuming a distributed unfair
daemon.

In~\cite{BlinFP14}, authors use the concept of labeling scheme
introduced by Korman {\em et al}~\cite{KKP10j} to design silent
self-stabilizing algorithms with bounded memory per process.  Using
their approach, they show that, every static task has a silent
self-stabilizing algorithm which converges within a linear number of
rounds in an arbitrary identified network, however no move complexity
is given.

To our knowledge, until now, only two works~\cite{CDV09j,DIJ18tr}
conciliate general schemes for stabilization and efficiency in both
moves and rounds. In~\cite{CDV09j}, Cournier {\em et al} propose a
general scheme for snap-stabilizing wave, henceforth non-silent,
algorithms in arbitrary connected and rooted networks. Using their
approach, one can obtain snap-stabilizing algorithms that execute each
wave in polynomial number of rounds and moves. In~\cite{DIJ18tr},
authors propose a general scheme to compute, in a linear number of
rounds, spanning directed treelike data structures on arbitrary
networks. They also exhibit polynomial upper bounds on its
stabilization time in moves holding for large classes of
instantiations of their scheme. Hence, our approach is complementary
to~\cite{DIJ18tr}.

\paragraph{Roadmap.}

The remainder of the paper is organized as follows. In the next
section, we present the computational model and basic definitions.  In
Section~\ref{sec:statement--truc}, we define the notion of {\em
  acyclic strategy} based on the notions of top-down and bottom-up
actions. In Section~\ref{sect:move}, we exhibit a polynomial upper
bound on the move complexity of algorithms that follow an acyclic
strategy.  In Section~\ref{sect:toy}, we propose a simple case
study. This example shows that our upper bound is tight, but in
contrast, the acyclic strategy is not restrictive enough as it allows
degenerated solutions where the stabilization time in rounds is in
$\Omega(n)$ where $n$ is the number of processes in the network. In
Section \ref{sect:rounds}, we show that any algorithm that follows an
acyclic strategy and whose actions are locally mutually exclusive
stabilizes in $O(\HG)$ rounds, where $\HG$ is the height of the
spanning forest; we also show how to add this latter property without
increasing the move complexity. In Section~\ref{sect:appli}, we review
several existing works where our method allows to trivially deduce
both correctness and stabilization time (both in terms of moves and
rounds).  Section \ref{sec:conclusions} is dedicated to concluding
remarks.


\section{Preliminaries}

We consider the {\em locally shared memory model with composite
  atomicity}~\cite{D74j} where processes communicate using locally
shared {\em variables}.

\subsection{Network}
\label{sec:network}

A {\em network} is made of a set of $n$ interconnected {\em
  processes}.  Communications are assumed to be bidirectional. Hence,
we model the topology of the network by a simple undirected graph $G =
(V,E)$, where $V$ is a set of processes and $E$ is a set of edges that
represents communication links, \textit{i.e.}, $\{p, q\}\in E$ means
that $p$ and $q$ can directly exchange information. In this latter
case, $p$ and $q$ are said to be {\em neighbors}.
For a process $p\in V$, we denote by $p.\neigh$ the set of its neighbors:
$p.\neigh = \{ q\in V \suchthat \{p, q\}\in E\}$. We also note $\DeltaG$ the
degree of $G$, namely $\DeltaG = \max \{ |p.\neigh| \suchthat p\in V\}$.

\subsection{Algorithm}
\label{sec:algorithm}

A {\em distributed algorithm} $\mathcal{A}$ is a collection of $n =
|V|$ {\em local algorithms}, each one operating on a single
process: $\mathcal{A} = \{\mathcal{A}(p) \suchthat p \in V \}$ where
each process $p$ is equipped with a local algorithm
$\mathcal{A}(p) = (Var_p, Actions_p)$:
\begin{itemize}
\item $Var_p$ is the finite set of variables of $p$,
\item $Actions_p$ is the finite set of {\em actions} (guarded commands).
\end{itemize}
Notice that $\mathcal{A}$ may not be uniform in the sense that some
local algorithm $\mathcal{A}(p)$ may be different from some other(s).
We identify each variable involved in Algorithm $\mathcal{A}$ by the
notation $p.x\in Var_p$, where $x$ is the \textit{name} of the
variable and $p\in V$ the process that holds it.
Each process $p$ runs its local algorithm $\mathcal{A}(p)$ by
atomically executing actions. If executed, an action of $p$ consists
of reading all variables of $p$ and its neighbors, and then writing
into a part of the \emph{writable} ({\em i.e.}, non-constant)
variables of $p$. Of course, in this case, the written values depend on
the last values read by $p$.
For a process $p\in V$, each action in $Actions_p$ is written as follows
$$L(p)\ ::\ G(p)\ \to\ S(p)$$
$L(p)$ is a {\em label} used to identify the action in the discussion.
The {\em guard} $G(p)$ is a Boolean predicate involving variables of
$p$ and its neighbors. The {\em statement} $S(p)$ is a sequence of
assignments on writable variables of $p$.  
A variable $q.x$ is said to be {\em $G$-read} by $L(p)$ if $q.x$ is
involved in predicate $G(p)$ (in this case, $q$ is either $p$ or one
of its neighbors). Let $\pr{L(p)}$ be the set of variables that are
$G$-read by $L(p)$.
A variable $p.x$ is said to be {\em written} by $L(p)$ if $p.x$
appears as a left operand in an assignment of $S(p)$. Let $\w{L(p)}$
be the set of variables written by $L(p)$.

An action can be executed by a process $p$ only if it is {\em
  enabled}, \textit{i.e.}, its guard evaluates to true.  By extension,
a process is said to be {\em enabled} when at least one of its actions
is enabled.

\subsection{Semantics}
\label{sec:semantics}

The \emph{state} of a process $p\in V$ is a vector of valuations of
its variables and belongs to $\mathcal{C}(p)$, the Cartesian product
of the sets of all possible valuations for each variables of $p$.
A \emph{configuration} of an algorithm $\mathcal{A}$ is a vector made
of a state of each process in $V$. We denote by $\mathcal{C} =
\Pi_{p\in V} \mathcal{C}(p)$ the set of all possible configuration (of
$\mathcal{A}$).
For any configuration $\gamma\in\mathcal{C}$, we denote by $\gamma(p)$
({\em resp.} $\gamma(p).x$) the state of process $p\in V$ (resp. the
value of the variable $x\in Var_p$ of process $p$) in configuration
$\gamma$.

The asynchronism of the system is modeled by an adversary, called the
{\em daemon}. Assume that the current configuration of the system is
$\gamma$.  If the set of enabled processes in $\gamma$ is empty, then
$\gamma$ is said to be {\em terminal}. Otherwise, a {\em step of
  $\mathcal{A}$} is performed as follows: the daemon selects a
non-empty subset $S$ of enabled processes in $\gamma$, and every
process $p$ in $S$ \emph{atomically} executes one of its action
enabled in $\gamma$, leading the system to a new configuration
$\gamma'$. The step (of $\mathcal{A}$) from $\gamma$ to $\gamma'$ is
noted $\gamma\step\gamma'$: $\step$ is the binary relation over
$\mathcal{C}$ defining all possible steps of $\mathcal{A}$ in $G$.
Precisely, in $\gamma\step\gamma'$, for every selected process $p$,
$\gamma'(p)$ is set according to the statement of the action executed
by $p$ based on the values it G-reads on $\gamma$, whereas $\gamma'(q)
= \gamma(q)$ for every non-selected process $q$.

An {\em execution\/} of $\mathcal{A}$ is a maximal sequence
$\gamma_0 \gamma_1 ... \gamma_i ...$ of configurations of
$\mathcal{C}$ such that $\gamma_{i-1}\step\gamma_i$ for all
$i>0$. The term ``maximal'' means that the execution is either
infinite, or ends at a {\em terminal\/} configuration.

Recall that executions are driven by a daemon. We define a daemon
$\mathcal D$ as a predicate over executions. An execution $e$ is then
said to be an {\em execution under the daemon $\mathcal D$} if $e$
satisfies $\mathcal D$.  In this paper, we assume that the daemon is
{\em distributed} and {\em unfair}. ``Distributed'' means that, unless
the configuration is terminal, the daemon selects at least one enabled
process (maybe more) at each step. ``Unfair'' means that there is no
fairness constraint, \textit{i.e.}, the daemon might never select a
process unless it is the only enabled one.

\subsection{Time Complexity}

We measure the time complexity of an algorithm using two notions: {\em
  rounds}~\cite{DIM93} and {\em moves}~\cite{D74j}.
The complexity in {\em rounds} evaluates the execution
time according to the speed of the slowest processes.
The definition of round uses the concept of {\em neutralization}: a
process~$v$ is \textit{neutralized} during a step~$\gamma_i\step
\gamma_{i+1}$, if~$v$ is enabled in~$\gamma_i$ but not in
configuration~$\gamma_{i+1}$, and it is not activated in the
step~$\gamma_i\step\gamma_{i+1}$.  Then, the rounds are inductively
defined as follows. The first round of an execution~$e =
\gamma_0\gamma_1...$ is its minimal prefix~$e'$ such that every
process that is enabled in~$\gamma_0$ either executes an action or is
neutralized during a step of~$e'$. If $e'$ is finite, then the second
round of $e$ is the first round of the suffix
$\gamma_t\gamma_{t+1}...$ of~$e$ starting from the last configuration
$ \gamma_t$ of $e'$, and so forth.
The complexity in {\em moves} captures the amount of computations an
algorithm needs. Indeed, we say that a process {\em moves} in
$\gamma_i\step\gamma_{i+1}$ when it executes an action in
$\gamma_i\step\gamma_{i+1}$.

\subsection{Silent Self-Stabilization and Stabilization Time}
\begin{definition}
[Silent Self-Stabilization~\cite{DGS99}]
\label{def:selfstabilization}
Let $\mathcal{A}$ be a distributed algorithm for a network $G$, $SP$ a
predicate over the configurations of $\mathcal{A}$, and $\mathcal D$ a
daemon.  $\mathcal{A}$ is {\em silent and self-stabili\-zing for~$SP$
  in $G$ under $\mathcal D$} if the following two conditions hold:
\begin{itemize}
\item Every execution of $\mathcal{A}$ under~$\mathcal D$ is finite, and
\item every terminal configuration of $\mathcal{A}$ satisfies $SP$.
\end{itemize}
In this case, every terminal (resp. non-terminal)
configuration is said to be {\em legitimate} {\em
  w.r.t.} $SP$, (resp. {\em illegitimate} {\em
  w.r.t.} $SP$).
\end{definition}
The {\em stabilization time} in rounds (resp. moves) of a silent
self-stabilizing algorithm is the maximum number of rounds
(resp. moves) over every execution possible under the considered
daemon (starting from any initial configuration) to reach a terminal
(legitimate) configuration.

\section{Algorithm with Acyclic Strategy}
\label{sec:statement--truc}

In this section, we define a class of algorithm, the distributed
algorithms that {\em follow an acyclic strategy} and we study their
correctness and time complexity.  Let $\mathcal{A}$ be a distributed
algorithm running on some network $G = (V, E)$.

\subsection{Variable Names}
\label{sec:variable-names}

We assume that every process is endowed with the same set of variables
and we denote by $Names$ the set of names of those variables, namely:
$Names = \{ x \suchthat p \in V \wedge p.x\in Var_p\}$.  We also
assume that for every name $x\in Names$, for all processes $p$ and
$q$, variables $p.x$ and $q.x$ have the same definition domain.
The set of names is partitioned into two subsets: $ConstNames$, the
set of constant names, and $VarNames = Names \setminus ConstNames$,
the set of writable variable names. A name $x$ is in $VarNames$ as
soon as there exists a process $p$ such that $p.x \in Var_p$ and $p.x$
is written by an action of its local algorithm $\mathcal{A}(p)$.
For every $c \in ConstNames$ and every process $p\in V$, $p.c$ is
never written by any action and it has a pre-defined constant value
(which may differ from one process to another, {\em e.g.}, $\neigh$,
the name of the neighborhood).

We assume that $\mathcal A$ is {\em well-formed}, {\em i.e.},
$VarNames$ can be partitioned into $k$ sets $Var_1,$ $...,$ $Var_k$
such that $\forall p \in V$, $\mathcal A(p)$ consists of exactly $k$
actions $\FA_1(p), ..., \FA_k(p)$ such that $\w{\FA_i(p)} = \{p.v
\suchthat v \in Var_i\}$, for all $i \in \{1, ..., k\}$.  Let $\FA_i =
\{ \FA_i(p) \suchthat p \in V\}$, for all $i \in \{1, ..., k\}$.
Every $A_i$ is called a {\em family (of actions)}. By definition,
$\FA_1, ..., \FA_k$ is a partition over all actions of $\mathcal{A}$,
henceforth called a {\em families' partition}.

\begin{remark}\label{rem:wf}
  Since $\mathcal A$ is assumed to be {\em well-formed}, there is
  exactly one action of $\mathcal{A}(p)$ where $p.v$ is written, for
  every process $p$ and every writable variable $p.v$ (of $p$).
\end{remark}

\subsection{Spanning Forest}
\label{sec:spanning-forest}

In this work, we assume that every process is endowed with constant
variables that define a spanning forest over the graph $G$. Precisely,
we assume the constant names $parent, children\in ConstNames$ such that
for every process $p\in V$, $p.parent$ and $p.children$ are preset
as follows:
\begin{itemize}
\item $p.parent\in p.\neigh\cup\{\bot\}$: $p.parent$ is either a
  neighbor of $p$ (its {\em parent} in the forest), or $\bot$. In this
  latter case, $p$ is called a {\em (tree) root}.

  Hence, the graph made of vertices $V$ and edges $\{ (p, p.parent)
  \suchthat p\in V \wedge p.parent\neq\bot \}$ is assumed to be a
  spanning forest of $G$.
\item $p.children\subseteq p.\neigh$: $p.children$ contains the
  neighbors of $p$ which are the {\em children} of $p$ in the forest,
  {\em i.e.}, for every $p, q\in V$, $p.parent = q \iff p\in
  q.children$.

  Notice that the latter constraint implies that the graph made of
  vertices $V$ and edges $\{ (q, p) \suchthat p\in V \wedge q\in
  p.children \}$ is also a spanning forest of $G$.

  If $p.children = \emptyset$, then $p$ is called a {\em leaf}.
\end{itemize}
Note that $p.\neigh \setminus (\{p.parent\} \cup p.children)$ may
not be empty.
The set of $p$'s {\em ancestors}, $Ancestors(p)$, is recursively
defined as follows:
\begin{itemize}
\item $Ancestors(p) = \{p\}$ {\bf if} $p$ is a root,
\item $Ancestors(p) = \{p\} \cup Ancestors(p.parent)$ {\bf otherwise}.
\end{itemize}
Similarly, the set of $p$'s {\em descendants}, $Descendants(p)$, is
recursively defined as follows:
\begin{itemize}
\item $Descendants(p) = \{p\}$ {\bf if} $p$ is a leaf,
\item $Descendants(p) = \{p\} \cup \bigcup_{q \in p.children}
  Descendants(q)$ {\bf otherwise}.
\end{itemize}

\subsection{Acyclic Strategy}
\label{sec:bottom-up-top-down}

Let $\FA_1, ..., \FA_k$ be the families' partition of $\mathcal
A$. $\FA_i$, with $i \in \{1, ..., k\}$, is said to be {\em correct-alone} if
for every process $p$ and every step $\gamma \mapsto \gamma'$ such
that $\FA_i(p)$ is executed in $\gamma \mapsto \gamma'$, if no
variable in $\pr{A_i(p)} \setminus \w{\FA_i(p)}$ is modified in
$\gamma \mapsto \gamma'$, then $\FA_i(p)$ is disabled in $\gamma'$.
Notice that if a variable in $\w{\FA_i(p)}$ is modified in $\gamma
\mapsto \gamma'$, then it is necessarily modified by $\FA_i(p)$, by
Remark~\ref{rem:wf}.

Let $\precA$ be a binary relation over the families of actions of
$\mathcal A$ such that for $i, j\in\{1, ..., k\}$, $\FA_j \precA
\FA_i$ if and only if $i \neq j$ and there exist two processes $p$ and
$q$ such that $q \in p.\neigh\cup\{p\}$ and $\w{\FA_j(p)} \cap \pr{\FA_i(q)} \neq \emptyset$. We
conveniently represent the relation $\precA$ by a directed graph $\GC$
called {\em Graph of actions' Causality} and 
defined as follows: 
$\GC = (\{\FA_1, ..., \FA_k\}, \{(\FA_j, \FA_i), \FA_j \precA
\FA_i\})$.

Intuitively, a family of actions $\FA_i$ is top-down if activations of
its corresponding actions are only propagated down in the forest, {\em
  i.e.}, when some process $q$ executes action $\FA_i(q)$, $\FA_i(q)$
can only activate $\FA_i$ at some of its children $p$, if any. In this
case, $\FA_i(q)$ writes to some variables G-read by $\FA_i(p)$, these
latter are usually G-read to be compared to variables written by
$\FA_i(p)$ itself. In other words, a variable G-read by $\FA_i(p)$ can
be written by $\FA_i(q)$ only if $q = p$ or $q = p.parent$.
Formally, a family of actions $\FA_i$ is said to be {\em top-down} if
for every process $p$
 and every $q.v \in \pr{\FA_i(p)}$, we have
$q.v\in\w{\FA_i(q)}  \Rightarrow q \in \{p,p.parent\}$.

Intuitively, a family of actions $\FA_i$ is bottom-up if activations
of its corresponding actions are only propagated up in the forest,
{\em i.e.} when some process $q$ executes action $\FA_i(q)$,
$\FA_i(q)$ can only activate $\FA_i$ at its parent $p$, if any. In
this case, $\FA_i(q)$ writes to some variables G-read by $\FA_i(p)$,
these latter are usually G-read to be compared to variables written by
$\FA_i(p)$ itself. In other words, a variable G-read by $\FA_i(p)$
can be written by $\FA_i(q)$ only if $q = p$ or $q \in p.children$.
Hence, a family $\FA_i$ is said to be {\em bottom-up} if for every
process $p$
 and every $q.v \in \pr{\FA_i(p)}$, we have 
$q.v\in\w{\FA_i(q)}  \Rightarrow q \in p.children \cup \{p\}$.

A distributed algorithm $\mathcal A$ follows an {\em acyclic strategy}
if it is well-formed, its graph of actions' causality $\GC$
is {\em acyclic}, and for every $\FA_i$ in its families' partition,
$\FA_i$ is correct-alone and either bottom-up or top-down.

\section{Move Complexity of Algorithms with Acyclic Strategy}
\label{sect:move}

In this section, we exhibit a polynomial upper bound on the move
complexity of any algorithm that follows an acyclic strategy.
Throughout this section, we consider a distributed algorithm $\mathcal
A$ which follows an {\em acyclic strategy} and runs on the network $G
= (V,E)$. We use the same notation as in the previous section, in
particular, we let $\FA_1, ..., \FA_k$ be the families' partition of
$\mathcal A$.

\subsection{Definitions}
Let $p$ be a process and $\FA_i, i\in\{1, ..., k\}$ a family of
actions.

We define the {\em impacting zone} of $p$ and $\FA_i$, noted
$Z(p,\FA_i)$, as follows:
\begin{itemize} 
\item $Z(p,\FA_i) =  Ancestors(p)$ 
  {\bf if} $\FA_i$ is top-down,
\item $Z(p,\FA_i) = Descendants(p)$ {\bf otherwise} ({\em i.e.},
  $\FA_i$ is bottom-up).
\end{itemize}

\begin{remark}\label{rem:z}
  By definition, we have $1 \leq |Z(p,\FA_i)| \leq n$. Moreover, if
  $\FA_i$ is top-down, then we have
  $1 \leq |Z(p,\FA_i)| \leq H + 1 \leq n$, where $\HG$ is the height
  of $G$, {\em i.e.}, the maximum among the heights\footnote{The
    height of $p$ in $G$ is 0 if $p$ is a leaf. Otherwise the height
    of $p$ in $G$ is equal to one plus the maximum among the heights
    of its children.}  of the roots of all trees of the forest
\end{remark}

We also define the quantity $\MG(\FA_i,p)$ as:
\begin{itemize} 
\item the \emph{level}\footnote{The level of $p$ in $G$ is the
    distance from $p$ to the root of its tree in $G$ (0 if $p$ is the
    root itself).} of $p$ in $G$ {\bf if} $\FA_i$ is top-down,
\item the \emph{height} of $p$
  in $G$ {\bf otherwise} ({\em i.e.}, $\FA_i$ is bottom-up).
\end{itemize}

\begin{remark}\label{rem:mg}
  By definition, we have $0 \leq \MG(\FA_i,p) \leq \HG$, where $\HG$
  is the height of $G$.
\end{remark}

We define 
$$
Others(\FA_i,p) = \{q \in p.\neigh \suchthat \exists \FA_j, i
\neq j \wedge \w{\FA_j(q)} \cap \pr{\FA_i(p)} \neq \emptyset\}
$$
the set of neighbors $q$ of $p$ that have actions other than $\FA_i(q)$ which write
variables that are G-read by $\FA_i(p)$. We also note:
$$
maxO(\FA_i) = \max ( \{|Others(\FA_i, p)| \suchthat p \in V\} \cup
\{maxO(\FA_j) \suchthat \FA_j \precA \FA_i)\} )
$$

\begin{remark}\label{rem:maxO}
  By definition, we have $maxO(\FA_i) \leq \DeltaG$.
  %
  Moreover, if $\forall p \in V$, $\forall i\in\{1, ..., k\}$,
  $Others(\FA_i, p)$ is empty, {\em i.e.}, no neighbor $q$ of $p$
  writes into a variable read by $\FA_i(p)$ using an action other than
  $\FA_i(q)$, then $\forall j\in \{1, ..., k\}$, $maxO(\FA_j) = 0$.
\end{remark}

\subsection{Stabilization Time in Moves}

\begin{lemma}\label{lem:moves}
  Let $\FA_i$ be a family of actions and $p$ be a process.  For every
  execution $e$ of the algorithm $\mathcal A$ on $G$, we have
  $$
  \#m(e,\FA_i,p) \leq \bigg(n . \big(1 + \DeltamGC . \big(1 +
  maxO(\FA_i)\big)\big)\bigg)^{\HGC(\FA_i)} . |Z(p,\FA_i)|
  $$
  where $\#m(e,\FA_i,p)$ is the number of times $p$ executes
  $\FA_i(p)$ in $e$, $\DeltamGC$ is the in-degree of
  $\GC$,\footnote{$\DeltamGC = \max \{ |\{\FA_j \suchthat
    \FA_j\precA\FA_i\}| \suchthat i\in\{1,
    ..., k\}\}$.} and $\HGC(\FA_i)$ is the height of $\FA_i$ in
  $\GC$.\footnote{The height of $\FA_i$ in $\GC$ is 0 if $\FA_i$ is a
    leaf of $\GC$. Otherwise, it is equal to one plus the maximum of
    the heights of the  $\FA_i$'s predecessors {\em w.r.t.}
    $\precA$.}
\end{lemma}

\begin{proof}
  Let $e = \gamma_0...\gamma_x...$ be any execution of
  $\mathcal A$ on $G$.  \\
  Let $K(\FA_i,p) = \MG(\FA_i,p) + (\HG + 1).\HGC(\FA_i)$.
  We proceed by induction on $K(\FA_i,p)$.

  \begin{description}
  \item[Base Case:] Assume $K(\FA_i,p) = 0$ for some family $\FA_i$
    and some process $p$.  By definition, $\HG \geq 0$, $\HGC(\FA_i)
    \geq 0$ and $\MG(\FA_i,p) \geq 0$. Hence, $K(\FA_i,p) = 0$ implies
    that $\HGC(\FA_i) = 0$ and $\MG(\FA_i,p) = 0$.  Since
    $\MG(\FA_i,p) = 0$, $Z(p,\FA_i) = \{p\}$.  So, since $\FA_i$ is
    top-down or bottom-up, for every $q.v \in \pr{\FA_i(p)}$,
    $q.v\in\w{\FA_i(q)} \Rightarrow q = p$.  Moreover, since
    $\HGC(\FA_i) = 0$, $\forall j \neq i$, $\FA_j \not\precA
    \FA_i$. So, for every $j \neq i$ and every $q \in
    p.\neigh\cup\{p\}$, $\w{\FA_j(p)} \cap \pr{\FA_i(q)} = \emptyset$.
    Hence, no action except $\FA_i(p)$ can modify a variable in
    $\pr{A_i(p)}$. Thus, $\#m(e,\FA_i,p) \leq 1$ since $\FA_i$ is
    correct-alone.

  \item[Induction Hypothesis:] Let $K \geq 0$. Assume that for every
    family $\FA_j$ and every process $q$ such that $K(\FA_j,q) \leq K$, we
    have 
    $$
    \#m(e,\FA_j,q) \leq \bigg(n . \big(1 + \DeltamGC . \big(1 +
    maxO(\FA_j)\big)\big)\bigg)^{\HGC(\FA_j)} . |Z(q,\FA_j)|
    $$

  \item[Induction Step:] Assume that for some family $\FA_i$ and some
    process $p$, $K(\FA_i,p) = K+1$.  If $\#m(e,\FA_i,p)$ equals 0 or
    1, then the result trivially holds.
    Assume now that $\#m(e,\FA_i,p) > 1$ and consider two consecutive
    executions of $\FA_i(p)$ in $e$, {\em i.e.}, there exist $x,y$
    such that $0 \leq x < y$, $\FA_i(p)$ is executed in both $\gamma_x
    \mapsto \gamma_{x+1}$ and $\gamma_y \mapsto \gamma_{y+1}$, but not
    in steps $\gamma_z \mapsto \gamma_{z+1}$ with $z \in \{x+1, ...,
    y-1\}$.
    Then, since $\FA_i$ is correct-alone, at least one variable in
    $\pr{\FA_i(p)}$ has to be modified by an action other than
    $\FA_i(p)$ in a step $\gamma_z \mapsto \gamma_{z+1}$ with $z \in
    \{x, ..., y-1\}$ so that $\FA_i(p)$ becomes enabled again. Namely,
    there are $j \in \{1, ..., k\}$ and $q \in V$ such that {\em (a)
      $j \neq i$ or $q \neq p$}, $\FA_j(q)$ is executed in a step
    $\gamma_z \mapsto \gamma_{z+1}$, and $\w{\FA_j(q)} \cap
    \pr{\FA_i(p)} \neq \emptyset$. Note also that, by definition, {\em
      (b) $q \in p.\neigh \cup \{p\}$}. Finally, by definitions of
    top-down and bottom-up, {\em (a)}, and {\em (b)}, $\FA_j(q)$
    satisfies: (1) $j \neq i \wedge q = p$, (2) $j = i \wedge q \in
    p.\neigh \cap Z(p,\FA_i)$, or (3) $j \neq i \wedge q \in
    p.\neigh$. In other words, at least one of the three
    following cases occurs:
    \begin{enumerate}
    \item {\em $p$ executes $\FA_j(p)$ in step $\gamma_z \mapsto \gamma_{z+1}$
      with $j \neq i$ and $\w{\FA_j(p)} \cap \pr{\FA_i(p)} \neq
      \emptyset$.}
      
      Consequently, $\FA_j \precA \FA_i$ and, so, $\HGC(\FA_j) <
      \HGC(\FA_i)$.  Moreover, $\MG(\FA_j,p) - \MG(\FA_i,p) \leq \HG$
      and $\HGC(\FA_j) < \HGC(\FA_i)$ imply $K(\FA_j,p) < K(\FA_i,p) =
      K+1$.  Hence, by induction hypothesis, we have:
      $$
      \#m(e,\FA_j,p) \leq \bigg(n . \big(1 + \DeltamGC . \big(1 +
      maxO(\FA_j)\big)\big)\bigg)^{\HGC(\FA_j)} . |Z(p,\FA_j)|
      $$

    \item  {\em There is $q \in p.\neigh \cap Z(p,\FA_i)$ such that $q$ executes $\FA_i(q)$ in step $\gamma_z \mapsto
      \gamma_{z+1}$ and $\w{\FA_i(q)} \cap \pr{\FA_i(p)} \neq
      \emptyset$.} 

Then,  $\MG(\FA_i,q) <
      \MG(\FA_i,p)$. 
      Since $\MG(\FA_i,q) < \MG(\FA_i,p)$, $K(\FA_i,q) < K(\FA_i,p) =
      K+1$ and, by induction hypothesis, we have: 
      $$
      \#m(e,\FA_i,q) \leq \bigg(n . \big(1 + \DeltamGC . \big(1 +
      maxO(\FA_i)\big)\big)\bigg)^{\HGC(\FA_i)} . |Z(q,\FA_i)|
      $$
      
    \item {\em A neighbor $q$ of $p$ executes an action $\FA_j(q)$ in step
      $\gamma_z \mapsto \gamma_{z+1}$, with $j \neq i$ and $\w{\FA_j(q)}
      \cap \pr{\FA_i(p)} \neq \emptyset$.}
      
      Consequently, $q \in Others(\FA_i,p)$ and $\FA_j \precA \FA_i$
      and, so, $\HGC(\FA_j) < \HGC(\FA_i)$.  Moreover, $\MG(\FA_j,q) -
      \MG(\FA_i,p) \leq \HG$ and $\HGC(\FA_j) < \HGC(\FA_i)$ imply
      $K(\FA_j,q) < K(\FA_i,p) = K+1$. Hence, by induction hypothesis,
      we have: 
      $$
      \#m(e,\FA_j,q) \leq \bigg(n . \big(1 + \DeltamGC . \big(1 +
      maxO(\FA_j)\big)\big)\bigg)^{\HGC(\FA_j)} . |Z(q,\FA_j)|
      $$
    \end{enumerate}
    (Notice that Cases 1 and 3 can only occur when $\HGC(\FA_i) > 0$.)

    We now bound the number of times each of the three above cases
    occur in the execution $e$.
    \begin{description}
    \item[Case 1:] By definition, there exist at most $\DeltamGC$
      predecessors $\FA_j$ of $\FA_i$ in $\GC$ (\textit{i.e.}, such
      that $\FA_j \precA \FA_i$). For each of them, we have
      $\HGC(\FA_j) < \HGC(\FA_i)$, $|Z(p,\FA_j)| \leq n$ (by
      Remark~\ref{rem:z}) and $maxO(\FA_j) \leq maxO(\FA_i)$. Hence,
      overall this case appears at most
      $\sum_{\{\FA_j \suchthat \FA_j\precA \FA_i\}} \#m(e,\FA_j,p)$
      \begin{align}
        \nonumber
        & \leq &  \sum_{\{\FA_j \suchthat \FA_j\precA \FA_i\}} 
        \bigg(n . \big(1 + \DeltamGC . \big(1 +
        maxO(\FA_j)\big)\big)\bigg)^{\HGC(\FA_j)} . |Z(p,\FA_j)|\\ 
        \nonumber
        & \leq &  \sum_{\{\FA_j \suchthat \FA_j\precA \FA_i\}} 
        \bigg(n . \big(1 + \DeltamGC . \big(1 +
        maxO(\FA_i)\big)\big)\bigg)^{\HGC(\FA_i)-1} . n\\ 
        \nonumber
        & \leq & \DeltamGC
        \bigg(n . \big(1 + \DeltamGC . \big(1 +
        maxO(\FA_i)\big)\big)\bigg)^{\HGC(\FA_i)-1} . n\\ 
        & \leq & \DeltamGC . n^{\HGC(\FA_i)} .
        \big(1 + \DeltamGC . \big(1 +
        maxO(\FA_i)\big)\big)^{\HGC(\FA_i)-1}
      \end{align}
    \item[Case 2:] By definition,
      $$Z(p,\FA_i) = \{p\} \uplus \biguplus_{q \in p.\neigh \cap Z(p,\FA_i)} Z(q,\FA_i)$$ Hence, overall this case appears
      at most
      $\sum_{q \in p.\neigh \cap Z(p,\FA_i)} \#m(e,\FA_i,q)$
      \begin{align}
        \nonumber 
        & \leq & \sum_{q \in p.\neigh \cap Z(p,\FA_i)} \bigg(n . \big(1 + \DeltamGC . \big(1 +
        maxO(\FA_i)\big)\big)\bigg)^{\HGC(\FA_i)} . |Z(q,\FA_i)| \\
        & \leq & n^{\HGC(\FA_i)} .  \big(1 + \DeltamGC . \big(1 +
        maxO(\FA_i)\big)\big)^{\HGC(\FA_i)} . \big(|Z(p,\FA_i)|
        -1\big)
      \end{align}
    \item[Case 3:] Again, for every $\FA_j \precA \FA_i$, we have
      $\HGC(\FA_j) < \HGC(\FA_i)$, $maxO(\FA_j) \leq maxO(\FA_i)$,
      and $Z(q,\FA_j) \leq n$ (Remark \ref{rem:z}). By
      definition, there are at most $\DeltamGC$ families $\FA_j$ such
      that $\FA_j \precA \FA_i$. Finally, $|Others(\FA_i,p)| \leq
      maxO(\FA_i)$, by definition. 
      Hence, overall this case appears at most

      $\sum_{\{\FA_j \suchthat \FA_j\precA \FA_i\}} \sum_{\{q \in Others(\FA_i,
        p)\}} \#m(e,\FA_j,q)$
      \begin{align} 
        \nonumber
        & \leq & \sum_{\{\FA_j \suchthat \FA_j\precA \FA_i\}} \sum_{\{q \in Others(\FA_i,
          p)\}}
        \bigg(n . \big(1 + \DeltamGC . \big(1 +
        maxO(\FA_j)\big)\big)\bigg)^{\HGC(\FA_j)} . |Z(q,\FA_j)| \\
        \nonumber
        & \leq & \sum_{\{\FA_j \suchthat \FA_j\precA \FA_i\}} \sum_{\{q \in Others(\FA_i,
          p)\}}
        \bigg(n . \big(1 + \DeltamGC . \big(1 +
        maxO(\FA_i)\big)\big)\bigg)^{\HGC(\FA_i)-1} . n\\
        & \leq &  \DeltamGC . maxO(\FA_i). n^{\HGC(\FA_i)} .
        \big(1 + \DeltamGC . \big(1 + 
        maxO(\FA_i)\big)\big)^{\HGC(\FA_i)-1}
      \end{align}
    \end{description}
    Overall $\#m(e,\FA_i,p)$ is less than or equal to 1 plus the sum
    of (1), (2), and (3) which is less than or equal to
    $$
    n^{\HGC(\FA_i)} .  \big(1 + \DeltamGC . \big(1 +
    maxO(\FA_i)\big)\big)^{\HGC(\FA_i)} . |Z(p,\FA_i)|
    $$
  \end{description}
\end{proof}

Since $maxO(\FA_i) \leq \DeltaG$ (Remark~\ref{rem:maxO}) and
$|Z(p,\FA_i)| \leq n$ (by Remark~\ref{rem:z}), we have

\begin{corollary}\label{coro:moves}
  Every execution of $\mathcal A$ on $G$ contains at most
  $\big(1 + \DeltamGC . (1+\DeltaG)\big)^{\HGC} . k . n^{\HGC+2}$
  moves, where $k$ is the number of families of $\mathcal A$, 
$\DeltamGC$ is the in-degree of
  $\GC$, and  $\HGC$ the height of $\GC$.
\end{corollary}

From Corollary \ref{coro:moves} and Definition~\ref{def:selfstabilization}, follows. 

\begin{theorem}\label{theo:silentself}
  Let $\mathcal{A}$ be a distributed algorithm for a network $G$
  endowed with a spanning forest, $SP$ a predicate over the
  configurations of $\mathcal{A}$.
  If $\mathcal{A}$ follows an {\em acyclic strategy} and every
  terminal configuration of $\mathcal{A}$ satisfies
  $SP$, then
\begin{itemize}
\item $\mathcal{A}$ is silent and self-stabili\-zing for~$SP$ in $G$
  under the distributed unfair daemon, and
\item its stabilization time
  is at most $\big(1 + \DeltamGC . (1+\DeltaG)\big)^{\HGC} . k
  . n^{\HGC+2}$ moves, 
\end{itemize}
where $k$ is the number of families of $\mathcal A$, $\DeltamGC$ is
the in-degree of $\GC$, and $\HGC$ the height of $\GC$.
\end{theorem}

\section{Toy Example}\label{sect:toy}

In this section, we propose a simple example of algorithm, called
Algorithm $\mathcal{TE}$, to show how to instantiate our results. The
aim of this section is threefold: (1) show that correctness and move
complexity of $\mathcal{TE}$ can be easily deduced from our general
results, (2) our upper bound on stabilization time in moves is tight
for this example, and (3) our definition of acyclic strategy allows
the design of solutions (like $\mathcal{TE}$) that are inefficient in
terms of rounds. We will show how to circumvent this latter negative
result in Section~\ref{sect:rounds}.

$\mathcal{TE}$ assumes a constant integer input $p.input \in \mathds N$ at
each process.  $\mathcal{TE}$ computes the sum of all inputs and then 
spreads this result everywhere in the network.
$\mathcal{TE}$ assumes that the network $T = (V,E)$ is a tree ({\em
  i.e.}, an undirected connected acyclic graph) with a sense of
direction (given by variables named $parent$ and $children$) which defines a
spanning in-tree rooted at process $r$ (the unique root, {\em i.e.},
the unique process satisfying $r.parent = \perp$).

Apart from those constant variables, every process $p$ has two
variables: $p.sub\in \mathds N$ (which is used to compute the sum of
input values in the subtree of $p$) and $p.res\in \mathds N$ (which
stabilizes to the result of the computation, {\em i.e.}, the sum of
all inputs).  The algorithm consists of two families of actions $S$
and $R$. $S$ computes variables $sub$ and is defined as follows.

For every process $p$ 
$$
S(p) ::\
p.sub \neq (\sum_{q \in p.children} q.sub) + p.input
\to p.sub \gets (\sum_{q \in p.children} q.sub) + p.input
$$

\noindent $R$ computes variables $res$ and is defined as follows.

$$R(r) :: r.res \neq r.sub \to r.res \gets r.sub$$
For every process $p \neq r$ 
$$R(p) :: p.res \neq \max(p.parent.res, p.sub) \to p.res \gets
\max(p.parent.res, p.sub)
$$

Remark that $S$ is bottom-up and correct-alone, while $R$ is top-down
and correct-alone. Moreover, the graph of actions' causality is
simply $$S \longrightarrow R$$
\noindent So, by Corollary~\ref{coro:moves} (with $\DeltamGC = 1$,
$\HGC = 1$ and $k = 2$), every execution of the algorithm contains at
most $(2 + \DeltaG).n^3$ moves and, as a direct consequence, every
execution terminates under the distributed unfair daemon. Notice also
that in every terminal configuration, every process $p$ satisfies the
following properties:
\begin{enumerate}
\item $p.sub = p.input + \sum_{q \in p.children} q.sub$,
\item $p.res = p.sub$ {\bf if} $p = r$, $p.res = \max(p.parent.res,
  p.sub)$ {\bf otherwise}.
\end{enumerate}
Let $P_{input} \equiv \forall p \in V, p.res = \sum_{q \in V}
q.input$.  By induction on the tree $T$, we can show that $P_{input}$
holds in any terminal configuration. Hence, by Theorem
\ref{theo:silentself}, follows:
\begin{lemma}
  The algorithm $\mathcal{TE}$ is silent and self-stabilizing for
  $P_{input}$ in $T$ under a distributed unfair daemon; its
  stabilization time is at most $(2+\DeltaG).n^3$ moves.
\end{lemma}

Using Lemma~\ref{lem:moves} directly, the move complexity of
$\mathcal{TE}$ can be further refined. Let $e$ be any execution and
$\HG$ be the height of $T$. First, note that, $maxO(S) = maxO(R) = 0$,
by Remark~\ref{rem:maxO}.

\begin{enumerate}
\item Since $S$ is bottom-up, $|Z(p,S)| \leq n$, for every process
  $p$. Moreover, the height of $S$ is 0 in the graph of actions'
  causality. Hence, by Lemma~\ref{lem:moves}, we have $\#m(e,S,p) \leq
  n$, for all processes $p$.  Thus, $e$ contains at most $n^2$
  moves of $S$.
\item Since $R$ is top-down, $|Z(p,R)| \leq
  \HG+1$, for every process $p$. Moreover, the height of $R$
  is 1 in the graph of actions' causality. Hence, by
  Lemma~\ref{lem:moves}, we have $\#m(e,R,p) \leq
  2.n.(\HG+1)$, for all processes $p$. Thus, $e$ contains at most
  $2.n^2.(\HG+1)$ moves of $R$.
\end{enumerate}
Overall, we have

\begin{lemma}
  \label{lem:stab-time-move-TE}
  The stabilization time of the algorithm $\mathcal{TE}$ is at most
  $n^2 (3 + 2\HG)$ moves, {\em i.e.}, $O(\HG.n^2)$ moves.
\end{lemma}

\subsection{Lower Bound in Moves}
\label{sec:lower-bound-moves}

We now show that the stabilization time of $\mathcal{TE}$ is
$\Omega(\HG.n^2)$ moves, meaning that the upper bound given by
Lemma~\ref{lem:stab-time-move-TE} is asymptotically reachable.
To that goal, we consider a directed line of $n$ processes, with
$n\geq 4$, noted $p_1, ..., p_n$: $p_1$ is the root and for every
$i\in\{2, ..., n\}$, there is a link between $p_{i-1}$ and $p_i$,
moreover, $p_i.parent = p_{i-1}$ (note that $\HG = n$).
We build a possible execution of $\mathcal{TE}$ running on this line that contains 
$\Omega(\HG.n^2)$ moves. We assume a central (unfair) daemon: at each
step exactly one process executes an action.  (The central daemon is a
particular case of the distributed unfair daemon.)

In this execution, we fix that $p_i.input = 1$, for every $i\in\{1, ..., n\}$. Moreover, we consider two classes of configurations: Configurations $X_{2i+1}$(with $3 \leq 2i+1\leq n$) and Configurations $Y_{2i+2}$ (with $4 \leq 2i+2\leq n$), see Figure~\ref{TE:conf-moves}. 

\begin{figure}[htpb]
  \textit{Configuration $X_{2i+1}$, $3 \leq 2i+1\leq n$:}

  \begin{center}
  \noindent$
  \begin{array}[t]{c|ccccccccccc}
    & p_1 & \dots & p_{2i-2} & p_{2i-1} & p_{2i} & p_{2i+1} & p_{2i+2}
    & p_{2i+3} & p_{2i+4} & p_{2i+5} & \dots \\
    \hline
    input & 1 & \dots & 1 & 1 & 1 & 1 & 1 & 1 & 1 & 1 & \dots \\
    sub & 2i & \dots & 3 & 2 & 1 & 0 & 2i & 0 & 2i+2 & 0 & \dots \\
    res & 2i & \dots & 2i & 2i & 2i & 0 & 0 & 0 & 0 & 0 & \dots
  \end{array}
  $
\end{center}

  \textit{Configuration $Y_{2i+2}$, $4 \leq 2i+2\leq n$:}
 
\begin{center}
  \noindent$
  \begin{array}[t]{c|ccccccccccc}
    & p_1 & \dots & p_{2i-2} & p_{2i-1} & p_{2i} & p_{2i+1} & p_{2i+2}
    & p_{2i+3} & p_{2i+4} & p_{2i+5} & \dots \\\hline
 input & 1 & \dots & 1 & 1 & 1 & 1 & 1 & 1 & 1 &  1& \dots \\
    sub & 4i+1 & \dots & 2i+4 & 2i+3 & 2i+2 & 2i+1 & 2i & 0 & 2i+2 &
    0 & \dots \\
    res & 4i+1 & \dots & 4i+1 & 4i+1 & 4i+1 & 4i+1 & 0 & 0 & 0 & 0 & \dots
  \end{array}
  $
\end{center}
  \caption{Configurations $X_{2i+1}$ and $Y_{2i+2}$}
  \label{TE:conf-moves}
\end{figure}

The initial configuration of the execution is $X_3$.  Then, we proceed
as follows: the system converges from configuration $X_{2i+1}$ to
configuration $Y_{2i+2}$ and then from $Y_{2i+2}$ to $X_{2i+3}$, back
and forth, until reaching a terminal configuration ($X_n$ if $n$ is
odd, $Y_n$ otherwise).

The system converges from
configuration $X_{2i+1}$ to configuration $Y_{2i+2}$, for every $i\geq
1$ and $2i+2\leq n$, in $\Omega(i^2)$ moves when the central daemon
activates processes in the following order:

\noindent\hfill\begin{minipage}[t]{.7\linewidth}
  \begin{algorithmic}[1]
    \DownFor{$j$}{$2i + 1$}{$1$} 
      \State $p_j$ executes $S(p_j)$ \Comment{$p_j.sub = 4i + 2 - j$} 
      \For{$k$}{$j$}{$2i + 1$} 
        \State $p_k$ executes $R(p_k)$ \Comment{$p_k.res = 4i + 2 - j$}
      \EndFor
    \EndDownFor
  \end{algorithmic}
\end{minipage}\hfill\hfill \\

Then, the system converges from configuration $Y_{2i+2}$ to
configuration $X_{2i+3}$, for every $i\geq 1$ and $2i+3\leq n$ in
$\Omega(i)$ moves when the central daemon activates processes in the
following order:

\noindent\hfill\begin{minipage}[t]{.7\linewidth}
  \begin{algorithmic}[1]
    \DownFor{$j$}{$2i + 2$}{$1$} 
       \State $p_j$ executes $S(p_j)$ \Comment{$p_j.sub = 2i + 3 - j$} 
    \EndDownFor
    \For{$j$}{$1$}{$2i + 2$} 
       \State $p_j$ executes $R(p_j)$ \Comment{$p_j.res = 2i + 2$}
    \EndFor
  \end{algorithmic}
\end{minipage}\hfill\hfill \\

Hence, following this scheduling of actions,  
the execution that starts in configuration $X_3$ converges to $X_n$ (resp. $Y_n$) if $n$ is odd
(resp. even) and  contains $\Omega(n^3)$ moves, {\em i.e.},
$\Omega(\HG.n^2)$ since the network is a line.

Remark that in this execution, for every process $p$, when $R(p)$ is
activated, $S(p)$ is disabled: this means that if the algorithm is
modified so that $S(p)$ has local priority over $R(p)$ for every
process $p$ (like in the method proposed in
Subsection~\ref{sub:trans}), the proposed execution is still possible
keeping a move complexity in $\Omega(\HG.n^2)$ even for such a
prioritized algorithm.





\subsection{Lower Bound in Rounds}\label{sub:rounds}

We now show that $\mathcal{TE}$ has a stabilization time in
$\Omega(n)$ rounds in any tree of height $\HG = 1$, {\em i.e.}, a star
network.  This negative result is mainly due to the fact that families
$R$ and $S$ are not locally mutually exclusive. In the next section,
we will propose a simple transformation to obtain a stabilization time
in $O(\HG)$ rounds, so $O(1)$ rounds in the case of a star network. We
will also show that this latter transformation does not affect the
move complexity.

Our proof consists in exhibiting a possible execution that terminates
in $n+3$ rounds assuming a central unfair daemon, that is, at
each step exactly one process executes an action. Notice that the
central unfair daemon is a particular case of the distributed unfair
daemon.


\begin{figure} [h]
\begin{minipage}{.28\linewidth}
  \begin{center}
    \includegraphics[scale=0.25]{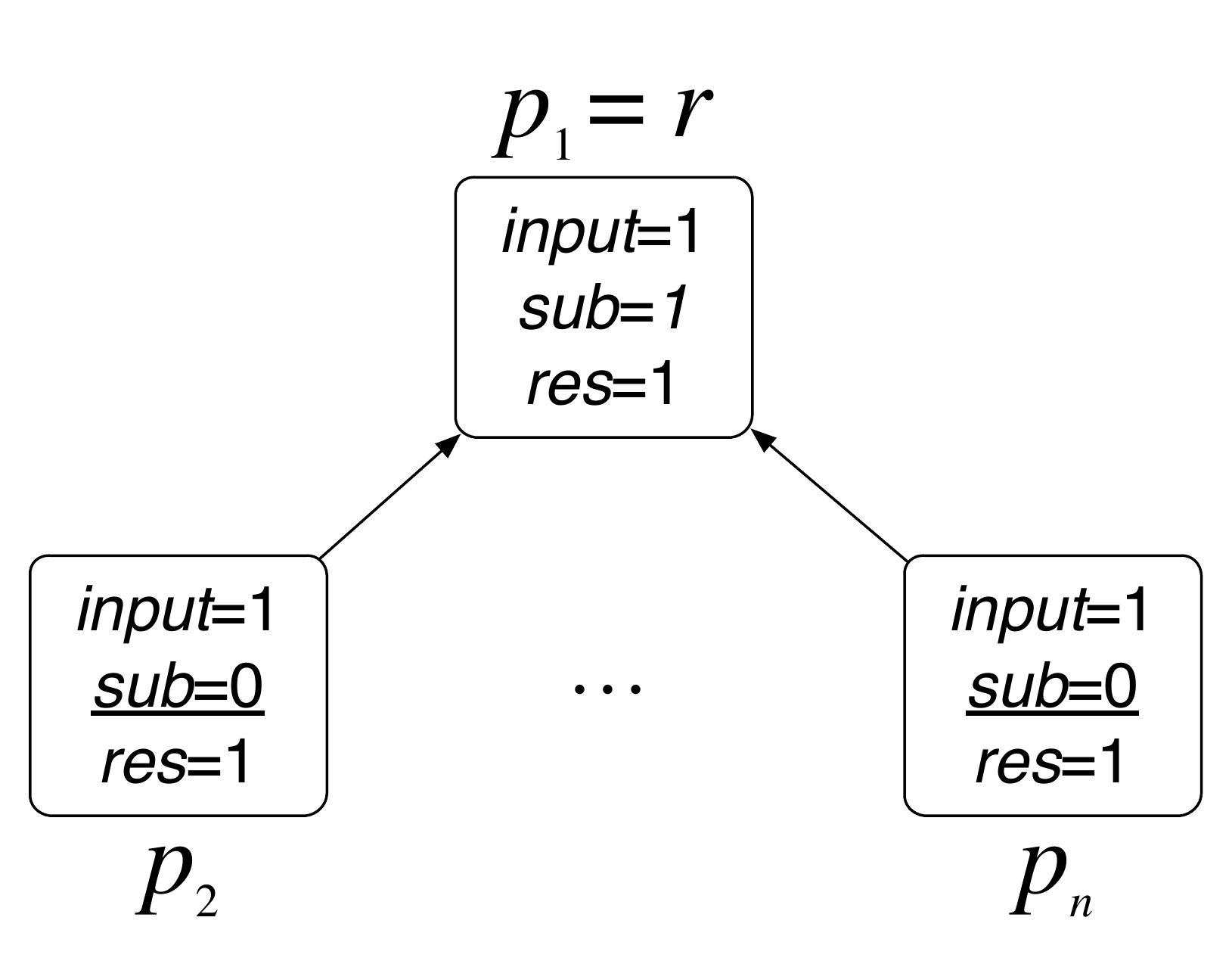}
 \caption{$C_1$, initial configuration.\label{fig:c1}}
 \end{center}
 \end{minipage}
\hfill
\begin{minipage}{.32\linewidth}
\vspace*{-15pt}
  \begin{center}
     \includegraphics[scale=.25]{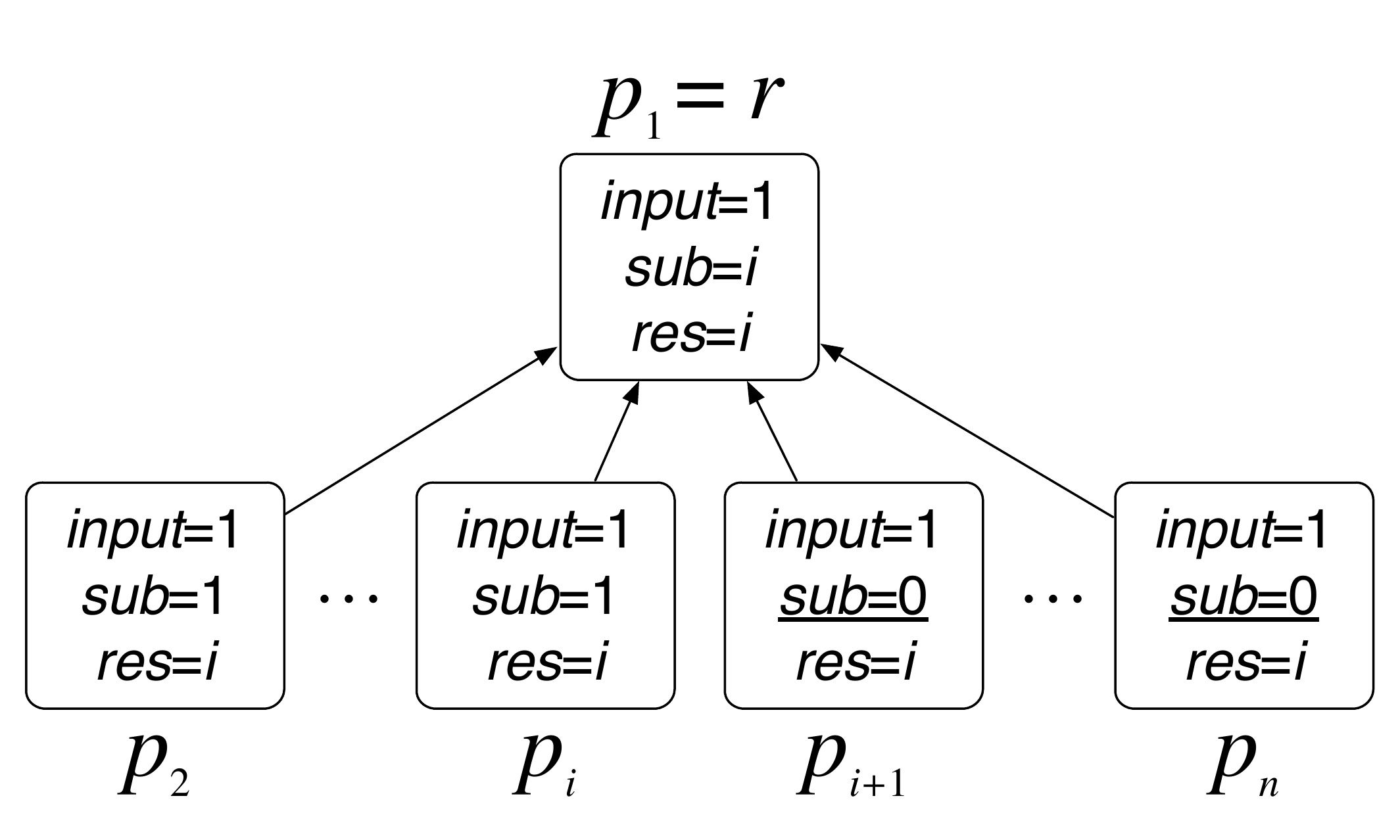} 
     \caption{$C_i$, $i \in \{2, ..., n-1\}$.\label{fig:ci}}
 \end{center}
\end{minipage}
\hfill
\begin{minipage}{.28\linewidth}
  \begin{center}
     \includegraphics[scale=.25]{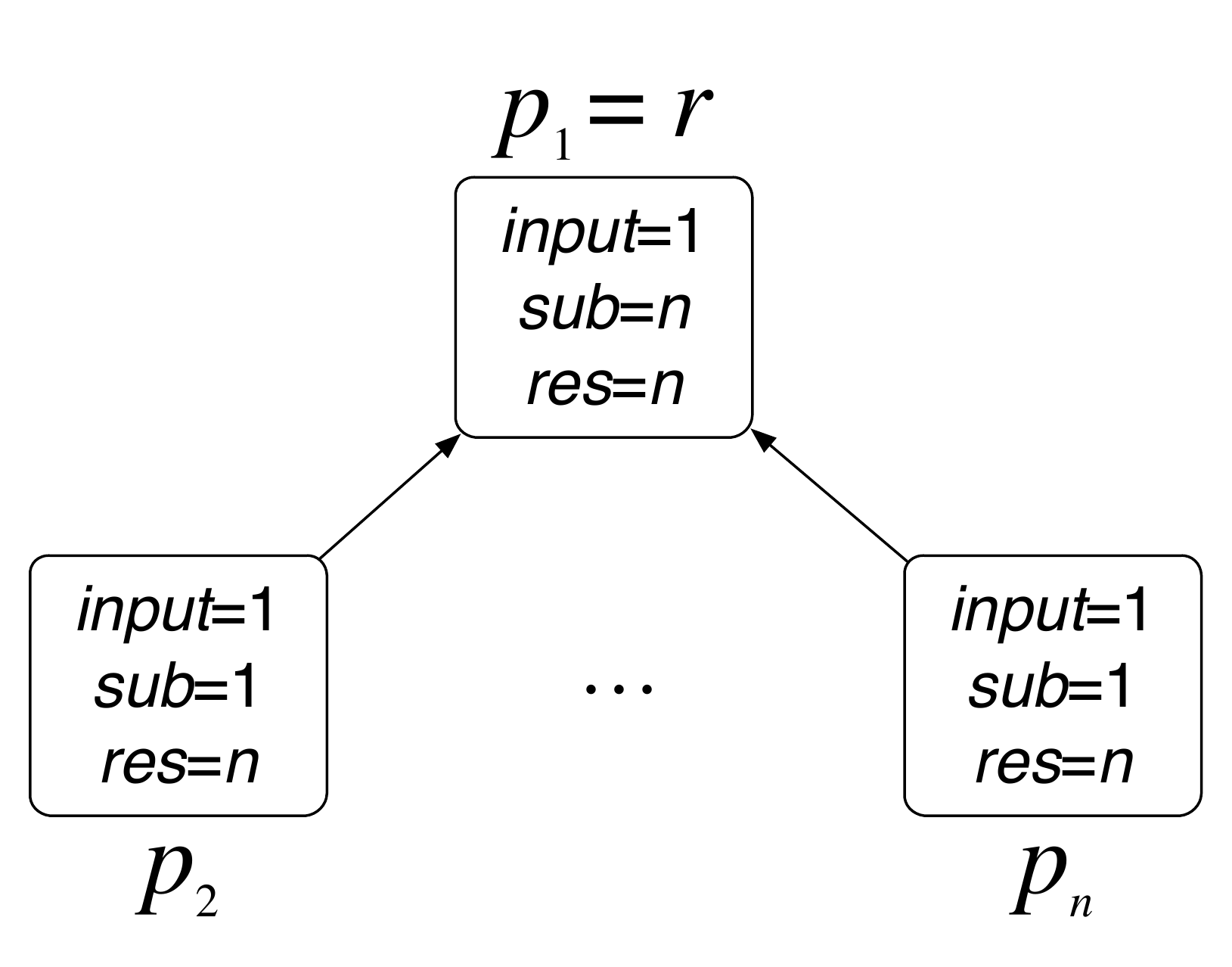} 
     \caption{$C_n$, terminal configuration.\label{fig:cn}}
 \end{center}
 \end{minipage}
\end{figure}

We consider a star network of $n$ processes ($n\geq 2$): $p_1$ is
the root of the tree and $p_2, ..., p_n$ are the leaves (namely links
are $\{\{p_1, p_i\}, i = 2, ..., n\}$). We note $C_i$, $i\in\{1, ...,
n\}$, the configuration satisfying the following three conditions:
\begin{itemize}
\item for every $j\in\{1, ..., n\}$, $p_j.input = 1$;
\item $p_1.sub = i$, for every $j\in\{2, ..., i\}$, $p_j.sub = 1$, and
  for every $j\in\{i+1, ..., n\}$, $p_j.sub = 0$; and
\item for every $j\in\{1, ..., n\}$, $p_j.res = i$.
\end{itemize}
$C_1$, $C_i$ with $i \in \{2, ..., n-1\}$, and $C_n$ are respectively
shown in Figures \ref{fig:c1}, \ref{fig:ci}, and \ref{fig:cn}. In
these figures, a variable is underlined whenever an action is enabled
to modify it.
Note that in configuration $C_i$, processes $p_1$, ..., $p_i$ are
disabled and processes $p_{i+1}, ..., p_n$ are enabled for $S$.  We
now build a possible execution that starts from $C_1$ and successively
converges to configurations $C_2$, ..., $C_n$ ($C_n$ is a terminal
configuration).  To converge from $C_i$ to $C_{i+1}$, $i\in\{1,
..., n-1\}$, the daemon applies the following scheduling:

\noindent\hfill\begin{minipage}[t]{.7\linewidth}
  \begin{algorithmic}[1]
    \State $p_{i+1}$ executes $S(p_{i+1})$ \Comment{$p_{i+1}.sub = 1$}
    \State $p_1$ executes $S(p_1)$ \Comment{$p_{1}.sub = i + 1$}
    \State $p_1$ executes $R(p_1)$ \Comment{$p_{1}.res = i + 1$}
    \For{$j$}{$2$}{$n$} 
        \State $p_j$ executes $R(p_j)$ \Comment{$p_j.res = i + 1$}
    \EndFor
  \end{algorithmic}
\end{minipage}\hfill\hfill~\\

For $i \in \{1, ..., n-2\}$, the convergence from $C_i$ to $C_{i+1}$
lasts exactly one round. Indeed, each process executes at least one
action between $C_i$ and $C_{i+1}$ and process $p_n$ is enabled at
configuration $C_i$ and remains continuously enabled until being
activated as the last process to execute in the round.
The convergence from $C_{n-1}$ to $C_n$ lasts four rounds: in
$C_{n-1}$, only $p_n$ is enabled to execute $S(p_n)$ hence the round
terminates in one step where only $S(p_n)$ is executed. Similarly,
$p_1$ then sequentially executes $S(p_1)$ and $R(p_1)$ in two
rounds. Finally, $p_2, ..., p_n$ execute $R$ in one round and then the
system is in the terminal configuration $C_n$.

Hence the above execution lasts $n+3$ rounds.

\section{Round Complexity of Algorithms with Acyclic Strategy}
\label{sect:rounds}

In this section, we first propose an extra condition that is
sufficient for any algorithm following an acyclic strategy to
stabilize in $O(\HG)$ rounds. We then propose a simple method to add
this property to any algorithm that follows an acyclic strategy,
without compromising the move complexity.

\subsection{A Condition for a Stabilization Time in $O(\HG)$ rounds}

Let $\FA_1, ..., \FA_k$ be the families' partition of $\mathcal A$.
We say that two families $\FA_i$ and $\FA_j$ are {\em locally mutually
  exclusive} if for every process $p$, there is no configuration
$\gamma$ where both $\FA_i(p)$ and $\FA_j(p)$ are enabled.
By extension, we say $\mathcal A$ is {\em locally mutually exclusive}
if for every $i,j \in \{1, ..., k\}$, $i \neq j$ implies that $\FA_i$
and $\FA_j$ are {\em locally mutually exclusive}.

\begin{theorem}\label{theo:rounds}
  Let $\mathcal{A}$ be a distributed algorithm for a network $G$
  endowed with a spanning forest.
  If $\mathcal{A}$ follows an acyclic strategy and is 
    locally mutually exclusive, then every execution of $\mathcal{A}$
  reaches a terminal configuration within at most $(\HGC+1).(\HG+1)$
  rounds, where $\HGC$ is the height of the graph of actions' causality
  $\GC$ of $\mathcal{A}$ and $\HG$ is the height of the spanning forest
  in $G$.
\end{theorem}
\begin{proof}
  Let $\FA_i$ be a family of actions of $\mathcal{A}$ and $p$ be a
  process. We note $R(\FA_i,p) = \HGC(\FA_i).(\HG+1) + \MG(\FA_i,p) +
  1$ (recall that $\HGC(\FA_i)$ and $\MG(\FA_i,p)$ are defined in
  Section~\ref{sect:move}).

  We now show by induction that for every family $\FA_i$ and every
  process $p$, after $R(\FA_i,p)$ rounds $\FA_i(p)$ is disabled
  forever.

  Let $p$ be a process and $\FA_i$ be a family. By definition,
  $\HGC(\FA_i) \geq 0$, $\HG \geq 0$, and $\MG(\FA_i,p) \geq 0$, hence
  $R(\FA_i,p) \geq 1$.
  \begin{description}
  \item[Base Case:] Assume that $R(\FA_i,p) = 1$. By definition,
    $\MG(\FA_i,p) = 0$ and $\HGC(\FA_i) = 0$.
Since
    $\MG(\FA_i,p) = 0$, $Z(p,\FA_i) = \{p\}$.  So, since $\FA_i$ is
    top-down or bottom-up, for every $q.v \in \pr{\FA_i(p)}$,
    $q.v\in\w{\FA_i(q)} \Rightarrow q = p$.  
    Moreover, since $\HGC(\FA_i) = 0$, $\forall j \neq i$, $\FA_j
    \not\precA \FA_i$. So, for every $j \neq i$ and every $q \in
    p.\neigh\cup\{p\}$, $\w{\FA_j(p)} \cap \pr{\FA_i(q)} = \emptyset$.
    Hence, no action except $\FA_i(p)$ can modify a variable in
    $\pr{A_i(p)}$.  Thus, if $\FA_i(p)$ is (initially) disabled, then
    $\FA_i(p)$ is disabled forever. Otherwise, $\FA_i(p)$ is
    continuously enabled until being executed; and, within at most
    one round, $\FA_i(p)$ is executed since $\mathcal{A}$ is locally
    mutually exclusive. After this first execution of $\FA_i(p)$,
    $\FA_i(p)$ is disabled forever since $\FA_i$ is correct-alone.

  \item[Induction Hypothesis:] Let $R \geq 1$. Assume that for every
    family $\FA_j$ and every process $q$ such that $R(\FA_j,q) \leq
    R$, after $R(\FA_j,q)$ rounds, $R(\FA_j,q)$ is disabled forever.
    
\item[Induction Step:] Assume that for some family $\FA_i$ and some
    process $p$, $R(\FA_i,p) = R+1$.

    Since $\FA_i$ is either bottom-up or top-down and by definition of
    $\precA$, we can deduce that for every family $\FA_j$, every $q
    \in p.\neigh \cup \{p\}$, and every $q.v \in \w{\FA_j(q)} \cap
    \pr{\FA_i(p)}$ one of the following four conditions holds:
\begin{enumerate}
\item {\em $j = i \wedge q = p$.}
\item {\em $j = i \wedge q \in p.\neigh \cap Z(p,\FA_i)$.} In this case, $\MG(\FA_i,q) <
      \MG(\FA_i,p)$, so $R(\FA_i,q) < R(\FA_i,p) = R+1$.
    \item {\em $j \neq i \wedge q = p \wedge \FA_j \precA \FA_i$.} In
      this case, $\FA_j \precA \FA_i$ implies that $\HGC(\FA_j) <
      \HGC(\FA_i)$, so $R(\FA_j,q) = R(\FA_j,p) < R(\FA_i,p) = R+1$.
\item {\em $j \neq i \wedge q \in p.\neigh \wedge \FA_j \precA
    \FA_i$.} In this case, $\FA_j \precA \FA_i$ implies that
  $\HGC(\FA_j) < \HGC(\FA_i)$. Moreover,
  $\MG(\FA_j,q) - \MG(\FA_i,p) \leq \HG$. So, $R(\FA_j,q) < R(\FA_i,p) =
  R+1$.
\end{enumerate}
Thus, by induction hypothesis, after $R$ rounds, 
all variables of
$\pr{\FA_i(p)}$ satisfying Cases (2), (3), or (4) are constant forever, {\em i.e.}, 
all variables of
$\pr{\FA_i(p)}$, except maybe those written by $\FA_i(p)$ itself (Case (1)), are
constant forever. So, if after $R$ rounds, $\FA_i(p)$ is disabled,
then it is disabled forever. Otherwise, after $R$ rounds, $\FA_i(p)$
is continuously 
enabled until being executed; and, within at most one additional round,
$\FA_i(p)$ is executed since $\mathcal{A}$ is locally mutually
exclusive. After the execution of $\FA_i(p)$, $\FA_i(p)$ is disabled
forever since $\FA_i$ is correct-alone.  Hence, after $R+1$ rounds,
$\FA_i(p)$ is disabled forever, and we are done.
  \end{description}
  Since for every family $\FA_i$ and every process $p$, $\HGC(\FA_i)
  \leq \HGC$ and $M(\FA_i, p) \leq \HG$, we have $R(\FA_i,p)
  \leq (\HGC+1).(\HG+1)$, hence the lemma holds.
\end{proof}

From Theorem~\ref{theo:rounds} and
Definition~\ref{def:selfstabilization}, follows.

\begin{corollary}\label{coro:rounds}
  Let $\mathcal{A}$ be a distributed algorithm for a network $G$
  endowed with a spanning forest and $SP$ a predicate over the
  configurations of $\mathcal{A}$. 
  If $\mathcal{A}$ follows an {\em acyclic strategy}, is {\em locally
    mutually exclusive}, and every terminal configuration of
  $\mathcal{A}$ satisfies  $SP$, then
  \begin{itemize}
  \item $\mathcal{A}$ is silent and self-stabili\-zing for~$SP$ in $G$
    under the distributed unfair daemon, and
  \item its stabilization time is at most $(\HGC+1).(\HG+1)$ rounds,
\end{itemize}
where $\HGC$ the height of the graph of actions' causality
  $\GC$ of $\mathcal{A}$ and $\HG$ is the
height of the spanning forest in $G$.
\end{corollary}

By definition, $\HGC < k$, the bound exhibited by the previous lemma
is in $O(k.\HG)$ where $k$ is the number of families of the algorithm.
Actually, the local mutual exclusion of the algorithm is usually implemented by
enforcing priorities on families as in the transformer presented below.
Hence, in practical cases, $\HGC = k-1$, as shown in Lemma~\ref{lem:k}.

\subsection{A Transformer}\label{sub:trans}

We have shown in Subsection~\ref{sub:rounds} that there exist
algorithms that follow an acyclic strategy, are not locally mutually
exclusive and stabilize in $\Omega(n)$ rounds in the worst case.  So,
we formalize now a generic method based on priorities over actions to
give the mutually exclusive property to such algorithms; this ensures
a complexity in $O(\HG)$ rounds. Notice that the method does not
degrade the move complexity.

Let $\mathcal{A}$ be any distributed algorithm for a network $G$
endowed with a spanning forest that follows an acyclic strategy. Let
$k$ be the number of families of $\mathcal{A}$.
In the following, for every process $p$ and every family $\FA_i$, we
identify the guard and the statement of Action $\FA_i(p)$ by $G_i(p)$
and $S_i(p)$, respectively.

Let $\prioA$ be any strict total order on families of $\mathcal A$
compatible with $\precA$, {\em i.e.}, $\prioA$ is a binary relation on
families of $\mathcal A$ that satisfies the following three
conditions:
\begin{description}
\item[Strict Order:] $\prioA$ is irreflexive and transitive;\footnote{Notice that irreflexivity
    and transitivity implies asymmetry.}
\item[Total:] for every two families $\FA_i, \FA_j$, we have either
  $\FA_i \prioA \FA_j$, $\FA_j \prioA \FA_i$, or $i = j$; and
\item[Compatibility:] for every two families $\FA_i, \FA_j$, if $\FA_i
  \precA \FA_j$, then $\FA_i \prioA \FA_j$.
\end{description}
Let $\trans{\mathcal A}$ be the following algorithm:
\begin{itemize}
\item $\trans{\mathcal A}$ and $\mathcal A$ have the same set of
  variables.
\item   
Every process $p \in V$ holds the following $k$ actions.  For every $i \in
  \{1, ..., k\}$, 

$$\FA^{\tt T}_i(p)\ ::\ G_i^{\tt T}(p)\ \to\ S_i^{\tt T}(p)$$

where $G_i^{\tt T}(p) = \big(\bigwedge_{\FA_j \prioA \FA_i}
\neg G_j(p)\big) \wedge G_i(p)$ and $S_i^{\tt T}(p) = S_i(p)$.

$G_i(p)$ (resp. the set $\{G_j(p)\suchthat \FA_j \prioA
\FA_i\}$) is called the {\em positive part} (resp. {\em negative
  part}) of $G^{\tt T}_i(p)$.
\end{itemize}
Notice that, by definition, $\precA$ is irreflexive and the graph of
actions' causality induced by $\precA$ is acyclic. Hence, there always
exists a strict total order compatible with $\precA$, {\em i.e.}, the
above transformation is always possible for any algorithm $\mathcal A$
which follows an acyclic strategy.

\begin{remark}
  $\trans{\mathcal A}$ is well-formed and $\FA^{\tt T}_1, ...,
  \FA^{\tt T}_k$ is the families' partition of $\trans{\mathcal A}$,
  where $\FA^{\tt T}_i = \{\FA^{\tt T}_i(p) \suchthat p \in V\}$, for
  every $i \in \{1, ..., k\}$.
\end{remark}

By construction, we have :
\begin{remark}\label{rem:cons}
  For every $i,j \in \{1, ..., k\}$ such that $i \neq j$, and every
  process $p$, the positive part of $G_j^{\tt T}(p)$ belongs to the
  negative part in $G_i^{\tt T}(p)$ if and only if $\FA_j \prioA
  \FA_i$.
\end{remark}

\begin{lemma}\label{lem:lme} $\trans{\mathcal A}$ is locally mutually
  exclusive.
\end{lemma}
\begin{proof}
  Let $\FA_i^{\tt T}$ and $\FA_j^{\tt T}$ be two different families of
  $\trans{\mathcal A}$. Then, either $\FA_i\prioA\FA_j$ or
  $\FA_j\prioA\FA_i$ ($\prioA$ is a strict total order). Without the
  loss of generality, assume $\FA_i\prioA\FA_j$. Let $p$ be any
  process and $\gamma$ be any configuration. The positive part of
  $G_i^{\tt T}(p)$ belongs to the negative part of $G_j^{\tt T}(p)$
  (see Remark \ref{rem:cons}), and consequently, $\FA_i^{\tt T}(p)$
  and $\FA_j^{\tt T}(p)$ cannot be both enabled in $\gamma$. Hence,
  $\FA_i^{\tt T}$ and $\FA_j^{\tt T}$ are locally mutually exclusive,
  which in turns implies that $\trans{\mathcal A}$ is locally mutually
  exclusive.
\end{proof}

\begin{lemma}\label{lem:ordre}
 For every $i,j \in
\{1, ..., k\}$, if $\FA^{\tt T}_j \prec_{\trans{\mathcal A}} \FA^{\tt
  T}_i$, then $\FA_j \prioA \FA_i$.
\end{lemma}
\begin{proof}
Let $\FA_i^{\tt T}$ and $\FA_j^{\tt T}$ be two families such that $\FA^{\tt T}_j \prec_{\trans{\mathcal
    A}} \FA^{\tt T}_i$.  Then, $i \neq j$ and there exist two
processes $p$ and $q$ such that $q \in p.\neigh \cup \{p\}$ and
$\w{\FA^{\tt T}_j(p)} \cap \pr{\FA^{\tt T}_i(q)} \neq
\emptyset$. Then, $\w{\FA^{\tt T}_j(p)} =\w{\FA_j(p)}$, and either
$\w{\FA_j(p)} \cap \pr{\FA_i(q)} \neq \emptyset$, or $\w{\FA_j(p)}
\cap \pr{\FA_k(q)} \neq \emptyset$ where $G_k(q)$ belongs to the
negative part of $G^{\tt T}_i(q)$. In the former case, we have
$\FA_j \precA \FA_i$, which implies that $\FA_j\prioA\FA_i$ ($\prioA$
is compatible with $\precA$). In the latter case, $\FA_j
\prec_{\mathcal A} \FA_k$ (by definition) and $\FA_k\prioA\FA_i$ (by
Remark~\ref{rem:cons}). Since, $\FA_j \prec_{\mathcal A} \FA_k$
implies $\FA_j\prioA\FA_k$ ($\prioA$ is compatible with $\precA$), by
transitivity we have $\FA_j \prioA \FA_i$.  Hence, for every $i,j \in
\{1, ..., k\}$, $\FA^{\tt T}_j \prec_{\trans{\mathcal A}} \FA^{\tt
  T}_i$ implies $\FA_j \prioA \FA_i$, and we are done.
\end{proof}

\begin{lemma}\label{lem:as}
 $\trans{\mathcal A}$ follows an acyclic strategy.
\end{lemma}
\begin{proof}

 Let $\FA_i^{\tt T}$ be a family
  of $\trans{\mathcal A}$. The lemma is immediate from the following
  three claims.

\begin{description}
\item[Claim I:] {\em $\FA_i^{\tt T}$ is correct-alone.}

  \noindent {\em Proof of the claim:} Since $\mathcal{A}$ follows an acyclic strategy, $\FA_i$ is correct-alone. Moreover,
  for every process $p$, we have $S_i^{\tt T}(p) = S_i(p)$ and $\neg G_i(p)
  \Rightarrow \neg G_i^{\tt T}(p)$. Hence, 
  $\FA_i^{\tt T}$ is also correct-alone.

\item[Claim II:] {\em $\FA_i^{\tt T}$ is either bottom-up or top-down.}

  \noindent {\em Proof of the claim:} Since $\mathcal{A}$ follows
  an acyclic strategy, $\FA_i$ is either bottom-up or top-down.
  Assume $\FA_i$ is bottom-up.  By construction, for every process
  $q$, $S_i^{\tt T}(q) = S_i(q)$, which implies that $\w{\FA^{\tt
      T}_i(q)} = \w{\FA_i(q)}$. Let $q.v \in \pr{\FA^{\tt
      T}_i(p)}$. 

  \begin{itemize}
\item Assume $q.v \in \pr{\FA_i(p)}$. Then $q.v\in\w{\FA_i(q)} \Rightarrow
  q \in p.children \cup \{p\}$ (since $\FA_i$ is bottom-up), {\em i.e.},
  $q.v\in\w{\FA^{\tt T}_i(q)} \Rightarrow q \in p.children \cup
  \{p\}$.

\item Assume now that $q.v \notin \pr{\FA_i(p)}$. Then
  $q.v \in \pr{\FA_j(p)}$ such that $G_j(p)$ belongs to the negative
  part of $G^{{\tt T}}_i(p)$, {\em i.e.}, $\FA_j \prioA \FA_i$ (Remark
  \ref{rem:cons}). Assume, by the contradiction, that
  $q.v \in\w{\FA^{\tt T}_i(q)}$. Then $q.v \in\w{\FA_i(q)}$, and since
  $p\in q.\neigh \cup \{q\}$ (indeed, $q.v \in \pr{\FA_j(p)}$), we
  have $\FA_i \precA \FA_j$. Now, as $\prioA$ is compatible with
  $\precA$, we have $\FA_i \prioA \FA_j$. Hence, $\FA_j \prioA \FA_i$
  and $\FA_i \prioA \FA_j$, a contradiction. Thus,
  $q.v \notin\w{\FA^{\tt T}_i(q)}$ which implies that
  $q.v\in\w{\FA^{\tt T}_i(q)} \Rightarrow q \in p.children \cup \{p\}$
  holds in this case.
\end{itemize}
Hence, $\FA_i^{\tt T}$ is bottom-up. 

Following a similar reasoning, if $\FA_i$ is top-down,
we can show $\FA_i^{\tt T}$ is top-down too.

\item[Claim III:] {\em The graph of
  actions' causality of $\trans{\mathcal A}$ is acyclic.}

\noindent {\em Proof of the claim:}
By Lemma~\ref{lem:ordre}, for every $i,j \in \{1, ..., k\}$, $\FA^{\tt
  T}_j \prec_{\trans{\mathcal A}} \FA^{\tt T}_i \Rightarrow \FA_j \prioA
\FA_i$. Now, $\prioA$ is a strict total order. So, the graph of
actions' causality of $\trans{\mathcal A}$ is acyclic.
\end{description}
\end{proof}

\begin{lemma}\label{sameexec}
  Every execution of $\trans{\mathcal A}$ is an execution of $\mathcal
  A$.
\end{lemma}
\begin{proof}
  The lemma is immediate from the following three claims.
\begin{description}
\item[Claim I:] {\em $\mathcal A$ and $\trans{\mathcal A}$ have the same set of
  configurations.}

\noindent {\em Proof of the claim:} By definition.

 \item[Claim II:] {\em Every step of $\trans{\mathcal A}$ is a step of $\mathcal
  A$. }

\noindent {\em Proof of the claim:} $G_i(p)$ is the guard of
$\FA_i(p)$ and the positive part of $G^{\tt T}_i(p)$. So, $G^{{\tt
    T}}_i(p)$ implies $G_i(p)$, {\em i.e.}, if $\FA^{\tt T}_i(p)$ is enabled,
then $\FA_i(p)$ is enabled. Since $S_i^{\tt T}(p) = S_i(p)$, we are
done.

 \item[Claim III:] {\em Let $\gamma$ be any configuration. $\gamma$ is terminal
  {\em w.r.t.}  $\trans{\mathcal A}$ if and only if $\gamma$ is
  terminal {\em w.r.t.}  $\mathcal A$.}

  \noindent {\em Proof of the claim:} 
 $\gamma$ is  terminal
  {\em w.r.t.}  $\trans{\mathcal A}$ if and only if
\begin{align*}
 \bigwedge_{p \in V}\bigwedge_{i \in \{1, ..., k\}} \neg G_i^{\tt T}(p)&=  \bigwedge_{p \in V}\bigwedge_{i \in \{1, ..., k\}} \neg \big(\big(\bigwedge_{\FA_j \prioA \FA_i} \neg G_j(p)\big) \wedge G_i(p)\big)\\
                                   &= \bigwedge_{p \in V}\bigwedge_{i \in \{1, ..., k\}} \big(\big(\bigvee_{\FA_j \prioA \FA_i} G_j(p)\big) \vee \neg G_i(p)\big)\\ 
                                   &= \bigwedge_{p \in V}\bigwedge_{i \in \{1, ..., k\}} \neg G_i(p)
\end{align*}
Now, $\bigwedge_{p \in V}\bigwedge_{i \in \{1, ..., k\}} \neg G_i(p)$ if
and only if $\gamma$ is terminal {\em w.r.t.}  $\mathcal A$.

  
\end{description}
\end{proof}

\begin{theorem}
   Let $\mathcal{A}$ be a distributed algorithm for a network $G$
  endowed with a spanning forest, $SP$ a predicate over the
  configurations of $\mathcal{A}$.
  If $\mathcal{A}$ follows an acyclic strategy, and is silent
  and self-stabili\-zing for~$SP$ in $G$ under the distributed unfair
  daemon, then
\begin{enumerate}
\item $\trans{\mathcal A}$ is silent and self-stabili\-zing for~$SP$
  in $G$ under the distributed unfair daemon, 
\item its stabilization time
  is  at most $(\HGC+1).(\HG+1)$ rounds, and
\item  its stabilization time in moves is less than or equal to
  the one of  $\mathcal{A}$.  
\end{enumerate}
 where $\HGC$ is the height of the graph of actions' causality
  $\GC$ of $\mathcal{A}$ and $\HG$ is the height of the spanning forest
  in $G$.
\end{theorem}
\begin{proof}
  {\em (1)} and {\em (3)} are consequences of Lemma~\ref{sameexec} and
  Corollary~\ref{coro:rounds}. {\em (2)} follows from Lemmas
  \ref{lem:lme}, \ref{lem:as}, and Theorem \ref{theo:rounds}.
\end{proof}

Using the above theorem, our toy example $\mathcal{TE}$ stabilizes in
at most $2(\HG+1)$ rounds, keeping a move complexity in
$\Theta(\HG.n^2)$ in the worst case (recall that the worst-case
execution of $\mathcal{TE}$ proposed in Subsection
\ref{sec:lower-bound-moves} is also a possible execution of
$\trans{\mathcal{TE}}$).

The next lemma shows that in usual cases, the height $\HGC$ of graph of
actions' causality of $\trans{\mathcal A}$ satisfies $\HGC = k-1$, where
$k$ is the number of families of $\mathcal A$.

\begin{lemma}\label{lem:k}
  If for every $i \in \{1, ..., k\}$ and every $p \in V$, $G_i(p)
  \not\equiv false$, then
\begin{center}
for every $x,y \in \{1, ..., k\}$,
  $\FA^{\tt T}_x \prec_{\trans{\mathcal A}} \FA^{\tt T}_y$ if and only
  if $\FA_x \prioA \FA_y$.
\end{center}
Consequently, the height $\HGC$ of graph
  of actions' causality of $\trans{\mathcal A}$ satisfies $\HGC = k-1$
   (indeed
  $\prioA$ is a strict total order).
\end{lemma}
\begin{proof}
  Let $x,y \in \{1, ..., k\}$.  By Lemma~\ref{lem:ordre}, $\FA^{\tt
    T}_x \prec_{\trans{\mathcal A}} \FA^{\tt T}_y \Rightarrow \FA_x
  \prioA \FA_y$.
  Assume now that $\FA_x \prioA \FA_y$.  By irreflexivity, $x \neq
  y$. Let $p$ be any process. Since $\FA_x$ is correct-alone and $G_x(p)
  \not\equiv false$, $\w{\FA_x(p)} \cap \pr{\FA_x(p)} \neq
  \emptyset$. Now, $G_x(p)$ belongs to the negative part of $\FA^{\tt
    T}_y$ (Remark~\ref{rem:cons}) and $\w{\FA^{{\tt T}}_x(p)} =
  \w{\FA_x(p)}$. So, $\w{\FA^{{\tt T}}_x(p)} \cap \pr{\FA^{\tt
      T}_y(p)} \neq \emptyset$. Hence, $\FA^{\tt T}_x
  \prec_{\trans{\mathcal A}} \FA^{\tt T}_y$, and we are done.
\end{proof}

\section{Related Work and Applications}\label{sect:appli}

In this section, we review some existing works from the literature and
show how to apply our generic results on them.
Those works propose silent self-stabilizing algorithms for directed
trees or network where a directed spanning tree is available. These algorithms are, or can be easily translated into, well-formed
algorithms that follow an acyclic strategy.
Hence, their correctness and time complexities (in moves and rounds)
are directly deduced from our results.

\paragraph{A Distributed Algorithm for Minimum Distance-$k$ Domination
  in Trees \cite{kdomset-turau}.}

This paper proposes three algorithms for directed trees. Each algorithm
is given with its proof of correctness and round complexity, however
move complexity is not considered.  Our results allow to obtain the
same round complexities, and additionally provide move complexities.

The first algorithm converges to a legitimate terminal configuration where a minimum distance-$k$ dominating
set is defined. 
%
This algorithm can be trivially translated in our model as an
algorithm with a single variable and a single action at each process
$p$,
$$
\FA_1(p) :: p.L \neq \mathcal{L}(p) \to p.L \gets \mathcal{L}(p)
$$
We do not explain here the algorithm, for the role of variable $L$ and
its computation using $\mathcal{L}(p)$, please refer to the original
paper \cite{kdomset-turau}.  Now, from the definition of $\mathcal{L}$
in \cite{kdomset-turau}, we know that $\mathcal{L}(p)$ depends on
$q.L$ for $q\in p.children$; hence the family $\FA_1$ is bottom-up and
correct-alone. Thus, we can deduce from our results that the
translation of this algorithm in our model is silent and
self-stabilizing with a stabilization time in $O(\HG)$ rounds
(Theorem~\ref{theo:rounds}) and $O(n^2)$ moves
(Theorem~\ref{theo:silentself}) where $\HG$ is the height of the tree
and $n$ is the number of processes.

The second algorithm is an extension of the first one since it computes both a minimum distance-$k$ dominating set
and a maximum distance-$2k$ independent set. This algorithm is made of two
families of actions $\FA_1$ and $\FA_2$: for every node $p$,
$$
\FA_1(p) :: p.L \neq \mathcal{L}(p) \to p.L \gets \mathcal{L}(p)
$$
$$
\FA_2(p) :: p.fading \neq fading(p) \to p.fading \gets fading(p)
$$
We already know that $\FA_1$ is bottom-up and correct-alone.  Then,
from the definitions given in \cite{kdomset-turau}, we can easily
deduce that $\FA_2$ is top-down and correct-alone since $fading(p)$
depends on $p.parent.fading$ and $q.L$ with $q \in p.children$,
which is not written by the family $\FA_2$. Hence, the graph
of actions' causality is $$\FA_1\longrightarrow\FA_2$$ Thus, we obtain a
stabilization time of $O(\HG)$ rounds (as in~\cite{kdomset-turau}),
but additionally we obtain a move complexity in $O(n^2.\HG)$.

The third algorithm computes minimum connected distance-$k$ dominating
sets using five families of actions $\FA_1, ..., \FA_5$:
$$
\FA_1(p) :: p.L \neq \mathcal{L}(p) \to p.L \gets \mathcal{L}(p)
$$
$$
\FA_2(p) :: p.level \neq level(p) \to p.level \gets level(p)
$$
$$
\FA_3(p) :: p.cds \neq cds'(p) \to p.cds \gets cds'(p)
$$
$$
\FA_4(p) :: p.cds \wedge p.dist_l \neq dist_l(p) \to p.dist_l \gets dist_l(p)
$$
$$
\FA_5(p) :: p.minc \neq minc(p) \to p.minc \gets minc(p)
$$
From \cite{kdomset-turau}:
\begin{itemize}
\item $\mathcal{L}(p)$ depends on $q.L$ for $q\in p.children$,
\item $level(p)$ depends on $p.parent.level$,
\item we note $cds'(p) = \texttt{if } p.L = k \texttt{ then } true
  \texttt{ else } cds(p)$ and $cds(p)$ depends on $q.cds$ for $q\in
  p.children$,
\item $dist_l(p)$ depends on $q.L$ and $q.cds$ for $q\in p.children$,
  and $p.parent.dist_l$,
\item $minc(p)$ depends on $p.level$, $q.cds$ and $q.minc$ for $q\in
  p.children$.
\end{itemize}
Hence $\FA_1, \FA_3, \FA_5$ are bottom-up and correct-alone and
$\FA_2, \FA_4$ are top-down and correct-alone. The graph of actions'
causality is acyclic since $\FA_1\prec\FA_3$, $\FA_1\prec\FA_4$,
$\FA_2\prec\FA_5$, $\FA_3\prec\FA_4$, and $\FA_3\prec\FA_5$; and its
height is $\HGC = 2$. Thus, conformly to~\cite{kdomset-turau}, we
obtain a round complexity in $O(\HG)$. Moreover,
Theorem~\ref{theo:silentself} provides a move complexity in
$O(\Delta^2.n^4)$ (with $\Delta$ the degree of the tree).

\paragraph{Self-stabilizing Tree Ranking
  \cite{chaudhuri-tree-ranking}.}

In this paper, the authors propose a silent self-stabilizing algorithm
that works on a directed tree and computes various rankings of the
processes following several kind of tree traversals such as pre-order
or breadth-first traversal, assuming a central unfair daemon. They
assume that each node knows a predefined order on its children so that
the traversal ordering is deterministic.

Following our method, the proposed algorithm is made of six families of
actions.
\begin{itemize}
\item $\FA_1$ computes $D$, the number of proper descendants of the
  process, it also copies the number of proper descendants of each of
  children of the process. This family is bottom-up and correct-alone.
\item $\FA_2$ computes $L$, the level of the node. This family is
  top-down and correct-alone.
\item $\FA_3$ computes $PRE$, the preorder rank of the node. The value
  of $PRE$ depends on values computed by $\FA_1$, so $\FA_3$ is
  top-down and correct-alone. $\FA_3$ also computes $LABEL$, which is
  an intermediate labelling used for breadth-first ranking, $LABEL$
  directly depends on the values of $L$ and $PRE$ of the process.
\item $\FA_4$ computes $POST$ and $RPOST$, the postorder and preorder
  ranks of the process. The values of $POST$ and $RPOST$ depend on
  values computed by $\FA_1$, so $\FA_4$ is top-down and
  correct-alone.
\item $\FA_5$ computes $DLIST$ (an intermediate list of nodes for
  breadth-first ranking) in a bottom-up and correct-alone manner,
  since the value of $DLIST$ depends on $LABEL$ at the node and
  $DLIST$ at the children of the process.
\item $\FA_6$ computes $RLIST$ (the list of all nodes in the
  breadth-first order) and $BFR$ the breadth-first rank of the
  process.  $\FA_6$ is top-down and correct-alone since the values
  written by $\FA_6$ depend on $RLIST$ at the process and its parent
  as well as $DLIST$ and $LABEL$ at the process.
\end{itemize}
In \cite{chaudhuri-tree-ranking}, the authors divide the algorithm in
two phases: in phase 1, actions $\FA_1$ and $\FA_2$ are executed and
converge and then, after global termination of phase 1, phase 2 begins
with the other actions from $\FA_3$ to $\FA_6$. But our results
apply: there is no need to separate those two phases and the full
algorithm is silent and self-stabilizing under an unfair distributed
daemon. Note that the fact that after termination, we have correct
tree rankings is proven in the original paper. Moreover, note that we
extend this result since it has been proven for a central unfair
daemon only.  For round complexity, we obtain $O(\HG)$ rounds, like in
the original paper. For move complexity, we obtain $O(\Delta^3.n^5)$
moves using an unfair distributed daemon, while the authors obtain
$O(n^2)$ moves using an unfair central daemon and assuming the
computation is divided in two separated phases. Note that the overhead
we obtain is not surprising since centrality and phase separation
remove any interleaving.

\paragraph{Improved Self-Stabilizing Algorithms for L(2, 1)-Labeling
  Tree Networks \cite{chaudhuri-labelling}.}

In \cite{chaudhuri-labelling}, the authors propose two silent
self-stabilizing algorithms for computing a particular labelling in
directed trees.  Although more simple, their solutions follows the
same ideas as in \cite{chaudhuri-tree-ranking}: each algorithm
contains a single family of actions which is correct-alone and
top-down. We obtain the same bounds as
in~\cite{chaudhuri-tree-ranking}, namely the two algorithms are silent
self-stabilizing under a distributed unfair daemon and converge within
$O(n.\HG)$ moves and $O(\HG)$ rounds.

\paragraph{An $O(n^2)$ Self-Stabilizing Algorithm for Computing
  Bridge-Connected Components \cite{chaudhuri-bridge}}

In \cite{chaudhuri-bridge}, an algorithm is proposed to compute
bridge-connected components in a network endowed with a  depth-first spanning
tree. The algorithm is
proven to be silent self-stabilizing under an unfair distributed
daemon. However, as in \cite{chaudhuri-tree-ranking}, it is separated into
two phases, the first phase has to be finished, globally, before the
second phase begins. The first phase corresponds to one
family of actions (that computes variable $S$) which are correct-alone
and bottom-up, while the second phase corresponds to a second family of
actions (to compute variable $BCC$) which is correct-alone and top-down. Our
results show the correctness of the algorithm without enforcing those
phases, with a stabilization time in $O(n^3)$ moves and $O(\HG)$
rounds, respectively. 
Note that the original paper does not provide the round complexity and
obtains $O(n^2)$ moves in case the two phases are executed in sequence
without any interleaving.

\paragraph{A Note on Self-Stabilizing Articulation Point Detection
  \cite{Chaudhuri-articulation}.}

This paper proposes a silent self-stabilizing solution for
articulation point detection  in a network endowed with a  depth-first spanning
tree. The algorithm is
exactly the first phase of \cite{chaudhuri-bridge}, {\em i.e.}, a
single family of correct-alone and bottom-up actions. It converges in
at most $O(n^2)$ moves and $O(\HG)$ rounds under an unfair distributed
daemon.

\paragraph{A Self-Stabilizing Algorithm for Finding Articulation
  Points \cite{Karaata-articulation}.}

The silent algorithm given in \cite{Karaata-articulation} finds
articulation points in a network endowed with a breadth-first spanning
tree, assuming a central unfair daemon. The algorithm computes for
each node $p$ the variable $p.e$ which contains every non-tree edges
incident on $p$ and some non-tree edges incident on descendants of $p$
once a terminal configuration has been reached. Precisely, a non-tree
edge $\{p,q\}$ is propagated up in the tree starting from $p$ and $q$
until the first common ancestor of $p$ and $q$. Based on $p.e$, the
node $p$ can decide whether or not
it is an articulation point.  The algorithm can be translated in our
model as a single family of actions which is correct-alone and
bottom-up. From our results, it follows that this algorithm is
actually silent and self-stabilizing even assuming a distributed
unfair daemon. Moreover, its stabilization time is in $O(n^2)$ moves
and $O(\HG)$ rounds, respectively.

\paragraph{A Self-Stabilizing Algorithm for Bridge Finding
  \cite{Karaata-Chaudhuri-bridge}.}

The algorithm in \cite{Karaata-Chaudhuri-bridge} computes bridges in a
network endowed with a breadth-first spanning
tree, assuming a distributed
unfair daemon. As in \cite{Karaata-articulation}, the algorithm
computes a variable $p.s$ at each node $p$ using a single family which
is correct-alone and bottom-up.  The correctness of this algorithm
assuming a distributed unfair daemon is direct from our results.
Moreover, we obtain a stabilization time in $O(n^2)$ moves and
$O(\HG)$ rounds.

\paragraph{A Silent Self-stabilizing Algorithm for Finding Cut-Nodes
  and Bridges \cite{devismes-bridges}.}

The algorithm in \cite{devismes-bridges} computes cut-nodes and
bridges on connected graph endowed with a depth-first spanning
tree. It is silent and self-stabilizing under a distributed unfair
distributed daemon and converges within $O(n^2)$ moves and $O(\HG)$
rounds, respectively. Indeed, the algorithm contains a single family
of actions which is correct-alone and bottom-up.

\section{Conclusion}
\label{sec:conclusions}

We have presented a general scheme to prove and analyze silent
self-stabilizing algorithms running on networks endowed with a sense
of direction describing a spanning forest. Our results allow to easily
({\em i.e.} quasi-syntactically) deduce upper bounds on move and round
complexities of such algorithms. We have shown, using a toy example, that
our method allow to easily obtain tight complexity bounds,
precisely a stabilization time which is asymptotically optimal in
rounds and polynomial in moves.
Finally, we reviewed a number of existing silent self-stabilizing
solutions from the
literature~\cite{kdomset-turau,chaudhuri-tree-ranking,chaudhuri-labelling,chaudhuri-bridge,Chaudhuri-articulation,Karaata-articulation,devismes-bridges}
where our method applies.  In some of them, we were able to provide
more general results than those proven in the original papers. Namely,
some algorithms are proven using a strong daemon, whereas our work
extends to the most general daemon assumption, {\em i.e.}, the
distributed unfair daemon. Moreover, many papers only focus on one
kind of time complexity measure, whereas our results systematically
provide round as well as move complexities.

In many of those related works, the assumption about the existence of
a directed (spanning) tree in the network has to be considered as an
intermediate assumption, since this structure has to be built by an
underlying algorithm.  Now, there are several silent self-stabilizing
spanning tree constructions that are efficient in both rounds and
moves, {\em e.g.},~\cite{DIJ18tr}. Thus, both algorithms, {\em i.e.},
the one that builds the tree and the one that computes on this tree,
have to be carefully composed to obtain a general composite
algorithm where, the stabilization time is keeped both asymptotically
optimal in rounds and polynomial in moves.

\clearpage
\bibliographystyle{unsrt}
\bibliography{biblio}



\end{document}